\pdfoutput=1
\documentclass[11pt,fleqn]{article}

\usepackage[utf8]{inputenc}
\usepackage[T1]{fontenc}

\usepackage{comment}
\usepackage{fullpage}
\usepackage[indent,parfill]{parskip}
\usepackage{indentfirst}

\usepackage{amsmath,amsthm}
\usepackage{amssymb,amsfonts}
\usepackage{mathtools,thm-restate}

\usepackage{mathptmx}
\usepackage{esvect}  
\usepackage{stmaryrd}  
\usepackage[cal=pxtx,scr=boondox]{mathalpha}

\usepackage[style=alphabetic,natbib=true,sortcites=true,backref=true,maxnames=99,useprefix=true,maxalphanames=99]{biblatex}

\DefineBibliographyStrings{english}{%
  backrefpage = {$\Lsh$~p.},
  backrefpages = {$\Lsh$~pp.},
}
\usepackage[colorlinks=true,citecolor=blue!65!black,linkcolor=red!75!black]{hyperref}
\usepackage[nameinlink]{cleveref}

\usepackage[basic,sanserif]{complexity}

\usepackage{booktabs,colortbl,multirow}
\usepackage{diagbox,enumitem}
\usepackage{framed,fancybox,ascmac}
\usepackage{subfig}
\usepackage{tablefootnote}
\usepackage{threeparttable}

\usepackage{comment}
\usepackage{url}

\usepackage[colorinlistoftodos,prependcaption]{todonotes}

\usepackage{xspace}
\usepackage{bm}
\usepackage{mleftright}
\mleftright

\crefname{theorem}{Theorem}{Theorems}
\crefname{proposition}{Proposition}{Propositions}
\crefname{lemma}{Lemma}{Lemmas}
\crefname{claim}{Claim}{Claims}
\crefname{corollary}{Corollary}{Corollaries}
\crefname{remark}{Remark}{Remarks}
\crefname{observation}{Observation}{Observations}
\crefname{hypothesis}{Hypothesis}{Hypotheses}
\crefname{fact}{Fact}{Facts}
\crefname{definition}{Definition}{Definitions}
\crefname{problem}{Problem}{Problems}
\crefname{example}{Example}{Examples}

\crefname{appendix}{Appendix}{Appendices}
\crefname{section}{Section}{Sections}
\crefname{equation}{Eq.}{Eqs.}
\crefname{figure}{Figure}{Figures}
\crefname{table}{Table}{Tables}
\crefname{algorithm}{Algorithm}{Algorithms}

\usepackage[ruled]{algorithm}
\usepackage[commentColor=blue]{algpseudocodex}

\renewcommand{\Return}{\textbf{return}\xspace}
\algrenewcommand\textproc{\textsl}

\renewcommand{\geq}{\geqslant}
\renewcommand{\leq}{\leqslant}
\renewcommand{\epsilon}{\varepsilon}
\renewcommand{\theta}{\vartheta}
\renewcommand{\phi}{\varphi}
\renewcommand{\bar}{\overline}

\renewcommand{\tilde}{\widetilde}

\newcommand{\prb}[1]{\textsc{#1}\xspace}

\newcommand{\sq}[1]{\vv{#1}}

\newcommand{\reco}{\leftrightsquigarrow}
\newcommand{\defeq}{\coloneq}

\let\E\relax\DeclareMathOperator*{\E}{\mathbb{E}}  
\let\Pr\relax\DeclareMathOperator*{\Pr}{\mathbb{Pr}}
\DeclareMathOperator{\bigO}{\mathcal{O}}
\DeclareMathOperator{\val}{\mathsf{val}}
\DeclareMathOperator{\opt}{\mathsf{opt}}

\newcommand{\sss}{\mathsf{start}}
\newcommand{\ttt}{\mathsf{end}}
\newcommand{\zo}{\{0,1\}}
\newcommand{\Yes}{\textsc{Yes}\xspace}
\newcommand{\No}{\textsc{No}\xspace}

\newcommand{\V}{\calV}
\renewcommand{\W}{\calW}
\newcommand{\X}{\calX}
\newcommand{\A}{\calA}
\newcommand{\pf}{\pi}
\newcommand{\qf}{\sigma}
\newcommand{\Pf}{\Pi}

\newcommand{\asg}{\alpha}
\newcommand{\bsg}{\beta}
\newcommand{\asgrnd}{\mathbf{A}}  

\newcommand{\stsqasg}{\scrA}

\newcommand{\gap}{g}
\newcommand{\rep}{\lambda}
\newcommand{\bal}{\mu}
\newcommand{\Vhorn}{\V_\mathrm{Horn}}
\newcommand{\scAND}{\textup{\texttt{AND}}\xspace}
\newcommand{\scOR}{\textup{\texttt{OR}}\xspace}
\newcommand{\scNOT}{\textup{\texttt{NOT}}\xspace}

\newcommand{\kSATReconf}{\prb{E$k$-SAT Reconfiguration}}
\newcommand{\MMkSATReconf}{\prb{Maxmin E$k$-SAT Reconfiguration}}
\newcommand{\treSATReconf}{\prb{E3-SAT Reconfiguration}}
\newcommand{\MMtreSATReconf}{\prb{Maxmin E3-SAT Reconfiguration}}

\newcommand{\calA}{\mathcal{A}}

\newcommand{\calI}{\mathcal{I}}

\newcommand{\calV}{\mathcal{V}}
\newcommand{\calW}{\mathcal{W}}
\newcommand{\calX}{\mathcal{X}}

\newcommand{\bbN}{\mathbb{N}}

\newcommand{\scrA}{\mathscr{A}}

\let\Pr\relax\DeclareMathOperator*{\Pr}{\mathbb{P}}

\newtheorem{theorem}{Theorem}[section]

\newtheorem{lemma}[theorem]{Lemma}
\newtheorem{claim}[theorem]{Claim}
\newtheorem{corollary}[theorem]{Corollary}

\newtheorem{fact}[theorem]{Fact}
\newtheorem{remark}[theorem]{Remark}
\theoremstyle{definition}
\newtheorem{definition}[theorem]{Definition}
\newtheorem{problem}[theorem]{Problem}
\newtheorem{example}[theorem]{Example}

\numberwithin{equation}{section}

\title{Asymptotically Optimal Inapproximability of \\ E$k$-SAT Reconfiguration}
\author{
Shuichi Hirahara \\
\small{National Institute of Informatics, Japan} \\
\small{\href{mailto:s_hirahara@nii.ac.jp}{\texttt{s\_hirahara@nii.ac.jp}}}
\and
Naoto Ohsaka \\
\small{CyberAgent, Inc., Japan} \\
\small{\href{mailto:ohsaka_naoto@cyberagent.co.jp}{\texttt{ohsaka\_naoto@cyberagent.co.jp}}}
}
\date{}

\begin{document}
\maketitle
\thispagestyle{empty}

\begin{abstract}In the \prb{Maxmin E$k$-SAT Reconfiguration} problem,
we are given a satisfiable $k$-CNF formula $\varphi$ where each clause contains exactly $k$ literals,
along with a pair of its satisfying assignments.
The objective is transform one satisfying assignment into the other
by repeatedly flipping the value of a single variable,
while maximizing the minimum fraction of satisfied clauses of $\varphi$ throughout the transformation.
In this paper, we demonstrate that the optimal approximation factor for \prb{Maxmin E$k$-SAT Reconfiguration}
is $1 - \Theta\left(\frac{1}{k}\right)$.
On the algorithmic side,
we develop a deterministic $\left(1-\frac{1}{k-1}-\frac{1}{k}\right)$-factor approximation algorithm for every $k \geq 3$.
On the hardness side, we show that
it is $\mathsf{PSPACE}$-hard to approximate this problem within a factor of $1-\frac{1}{10k}$
for every sufficiently large $k$.
Note that
an ``$\mathsf{NP}$ analogue'' of \prb{Maxmin E$k$-SAT Reconfiguration} is \prb{Max E$k$-SAT},
whose approximation threshold is $1-\frac{1}{2^k}$ shown by H\r{a}stad (JACM 2001).
To the best of our knowledge,
this is the first reconfiguration problem whose approximation threshold is
(asymptotically) \emph{worse} than that of its $\mathsf{NP}$ analogue.
To prove the hardness result,
we introduce a new ``non-monotone'' test,
which is specially tailored to reconfiguration problems, despite not being helpful in the PCP regime.
\end{abstract}

\section{Introduction}

\kSATReconf \cite{gopalan2009connectivity} is a canonical reconfiguration problem, defined as follows:
Let $\phi$ be a satisfiable \emph{E$k$-CNF formula}, where each clause contains exactly $k$ literals, over $n$ variables.
A sequence over assignments for $\phi$, denoted by $\sq{\asg} = (\asg^{(1)}, \ldots, \asg^{(T)})$, is
called a \emph{reconfiguration sequence} if
every adjacent pair of assignments $\asg^{(t)}$ and $\asg^{(t+1)}$ differ in a single variable.
In the \kSATReconf problem,
for a pair of satisfying assignments $\asg_\sss$ and $\asg_\ttt$ for $\phi$,
we are asked to decide if there exists a reconfiguration sequence $\sq{\asg}$ from $\asg_\sss$ to $\asg_\ttt$
consisting only of satisfying assignments for $\phi$.
In other words, \kSATReconf asks the $st$-connectivity question over
the \emph{solution space} of $\phi$,
which is the subgraph $G_\phi$ of the $n$-dimensional Boolean hypercube
induced by all satisfying assignments for $\phi$.
Studying \kSATReconf and its variants was originally motivated by the application to analyze
the structure of the solution space for Boolean formulas.
For a random instance $\phi$ of \prb{E$k$-SAT} (in a low-density regime),
the solution space $G_\phi$ breaks down into exponentially many ``clusters'' 
\cite{achlioptas2011solution,mezard2005clustering},
providing insight into the (empirical) performance of SAT solvers,
such as DPLL~\cite{achlioptas2004exponential} and Survey Propagation~\cite{mezard2002analytic}.
To shed light on the structure of the solution space in the \emph{worst case} scenario,
\citet[Theorem~2.9]{gopalan2009connectivity}
established a dichotomy theorem that classifies
the complexity of every reconfiguration problem over Boolean formulas
as $\cP$ or $\PSPACE$-complete; e.g.,
\kSATReconf is in $\cP$ if $k \leq 2$ and is $\PSPACE$-complete for every $k \geq 3$.
Moreover, the diameter of the connected components of $G_\phi$
can be exponential in the $\PSPACE$-complete case while
it is always linear in the $\cP$ case \cite[Theorem~2.10]{gopalan2009connectivity}.
See \cref{sec:related} for related work on other reconfiguration problems.

In this paper, we study \emph{approximability} of \kSATReconf.
Recently, approximability of reconfiguration problems has been studied from both hardness and algorithmic sides 
\cite{ohsaka2023gap,ohsaka2024gap,ohsaka2024alphabet,ohsaka2025approximate,hirahara2024probabilistically,hirahara2024optimal,karthikc.s.2023inapproximability,ohsaka2024tight,hirahara2025asymptotically,ohsaka2025yet}
(see also \cref{subsec:related:approx-reconf}).
In the approximate version of \kSATReconf, called \MMkSATReconf \cite{ito2011complexity},
for a satisfiable E$k$-CNF formula $\phi$ and
a pair of its satisfying assignments $\asg_\sss$ and $\asg_\ttt$,
we are asked to construct a reconfiguration sequence $\sq{\asg}$ from $\asg_\sss$ to $\asg_\ttt$
consisting of any (not necessarily satisfying) assignments for $\phi$.
The objective is to maximize the \emph{minimum} fraction of satisfied clauses of $\phi$,
where the minimum is taken over all assignments in $\sq{\asg}$.
Note that an ``$\NP$ analogue'' of \MMkSATReconf is \prb{Max E$k$-SAT}.
\begin{itembox}[l]{\MMkSATReconf}
\begin{tabular}{ll}
    \textbf{Input:}
    & a satisfiable E$k$-CNF formula $\phi$ and
    a pair of its satisfying assignments $\asg_\sss$ and $\asg_\ttt$.
    \\
    \textbf{Output:}
    & a reconfiguration sequence $\sq{\asg}$ from $\asg_\sss$ to $\asg_\ttt$.
    \\
    \textbf{Goal:}
    & maximize the minimum fraction of satisfied clauses of $\phi$ over all assignments in $\sq{\asg}$.
\end{tabular}
\end{itembox}

\noindent
Solving this problem, we may be able to find a ``reasonable'' reconfiguration sequence
consisting of almost-satisfying assignments, 
so that we can mange \No instances of \kSATReconf.
An example of \MMtreSATReconf is described as follows.

\begin{figure}[t]
    \centering
    \includegraphics[width=0.45\linewidth]{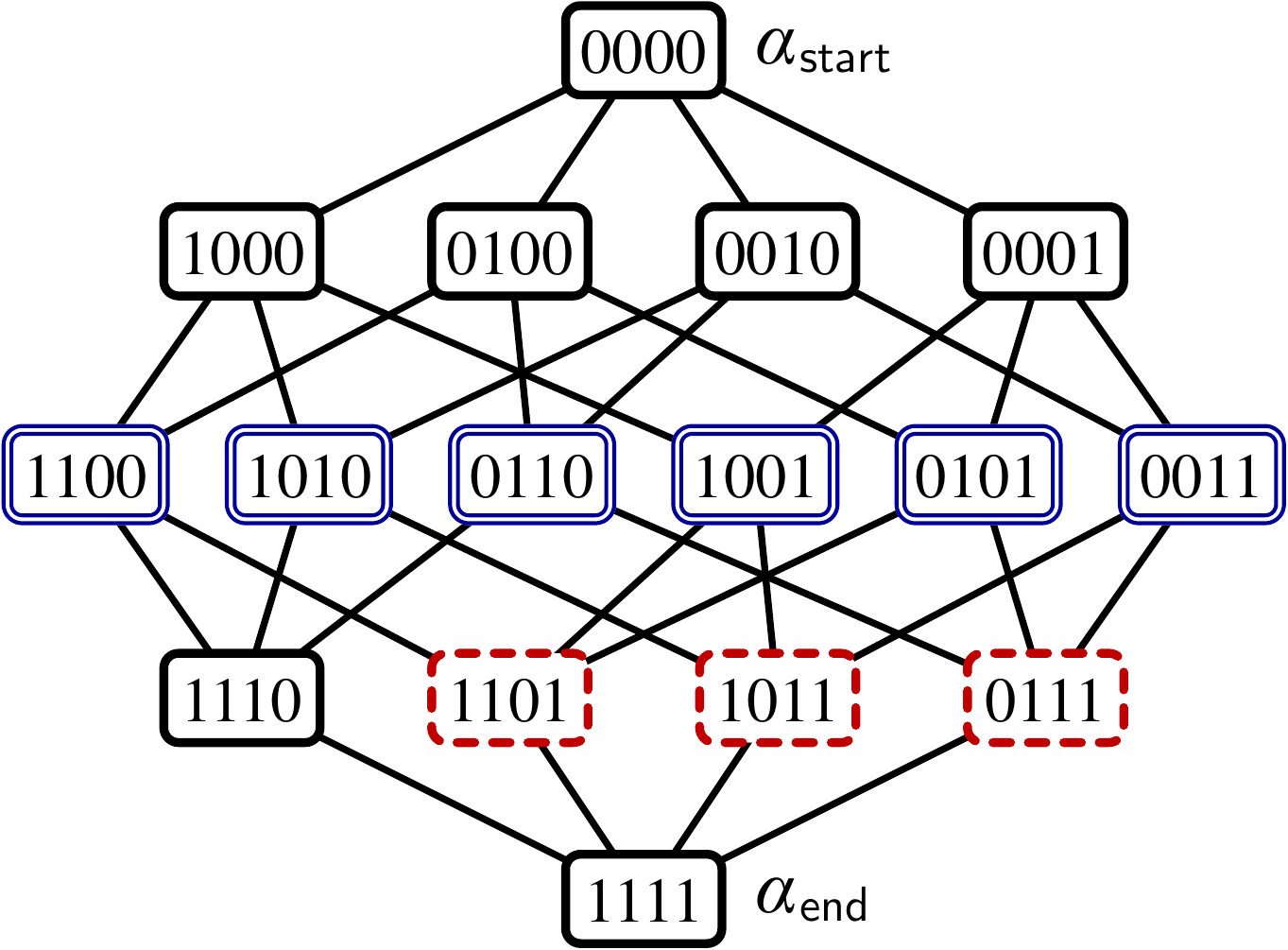}
    \caption{
        The solution space of \cref{ex}.
        Each assignment enclosed by a (blue) double line violates a single clause of an E$3$-CNF formula $\phi$,
        that enclosed by a (red) dashed line violates two clauses, and
        the other assignments satisfy $\phi$.
        Observe that we cannot transform $\asg_\sss$ into $\asg_\ttt$ without unsatisfying $\phi$;
        i.e., this is a \No instance of \treSATReconf.
        As an instance of \MMtreSATReconf,
        an optimal reconfiguration sequence is
        $(0000, 1000, 1100, 1110, 1111)$,
        whose objective value is $\frac{5}{6}$.
    }
    \label{fig:ex}
\end{figure}

\begin{oframed}
\begin{example}[\MMtreSATReconf]
\label{ex}
Let $\phi$
be an E$3$-CNF formula consisting of the following six clauses over four variables 
$x_1$, $x_2$, $x_3$, and $x_4$:
\begin{align}
\label{eq:ex}
\begin{aligned}
    C_1 & \defeq \bar{x_1}  \vee \bar{x_2} \vee x_3, &\quad
    C_4 & \defeq \bar{x_1} \vee x_2 \vee \bar{x_4}, \\
    C_2 & \defeq \bar{x_1}  \vee x_2 \vee \bar{x_3}, &\quad
    C_5 & \defeq \bar{x_2} \vee x_3 \vee \bar{x_4}, \\
    C_3 & \defeq x_1 \vee \bar{x_2} \vee \bar{x_3}, &\quad
    C_6 & \defeq x_1 \vee \bar{x_3} \vee \bar{x_4}.
\end{aligned}
\end{align}
Let $\asg_\sss \defeq 0000$ and $\asg_\ttt \defeq 1111$
be two satisfying assignments for $\phi$.
See \cref{fig:ex} for the solution space of $\phi$.
Observe that
$(\phi,\asg_\sss,\asg_\ttt)$ is a \No instance of \kSATReconf because
any reconfiguration sequence from $\asg_\sss$ to $\asg_\ttt$
passes through an assignment with exactly two $1$'s, which must violate one of the six clauses of $\phi$.
As an instance of \MMtreSATReconf,
any reconfiguration sequence from $\asg_\sss$ to $\asg_\ttt$ is considered feasible; e.g.,
$\sq{\asg} \defeq (0000, 0001, 0011, 0111, 1111)$
has the objective value $\frac{4}{6}$
since the fourth assignment $0111$ does not satisfy $C_3$ and $C_6$.
An optimal reconfiguration sequence is
$\sq{\asg}^* \defeq (0000, 1000, 1100, 1110, 1111)$,
whose objective value is $\frac{5}{6}$.
\end{example}
\end{oframed}

We review known results on the complexity of \MMkSATReconf.
For every $k \geq 3$, exactly solving \MMkSATReconf is $\PSPACE$-hard,
which follows from that of \kSATReconf \cite[Theorem~2.9]{gopalan2009connectivity}.
\citet[Theorem~5]{ito2011complexity} showed that \prb{Maxmin E$5$-SAT Reconfiguration}
is $\NP$-hard to approximate within a factor better than $\frac{15}{16}$.
For $\PSPACE$-hardness of approximation,
the \emph{Probabilistically Checkable Reconfiguration Proof} (PCRP) theorem
due to \citet[Theorem~1.5]{hirahara2024probabilistically} and \citet[Theorem~1]{karthikc.s.2023inapproximability},
along with a series of gap-preserving reductions due to \citet{ohsaka2023gap},
implies that
\MMtreSATReconf and \prb{Maxmin E2-SAT Reconfiguration} are
$\PSPACE$-hard to approximate within some constant factor.
So far, the \emph{asymptotic} behavior of approximability for \MMkSATReconf with respect to
the clause width $k$ is not well understood.

\subsection{Our Results}
\label{subsec:intro:results}
In this paper, we demonstrate that
the approximation threshold of \MMkSATReconf is $1 - \Theta\left(\frac{1}{k}\right)$.
On the algorithmic side,
we develop a deterministic $\left(1-\frac{1}{k-1}-\frac{1}{k}\right)$-factor approximation algorithm for every $k \geq 3$.

\begin{theorem}[informal; see \cref{thm:alg}]
\label{thm:intro:alg}
For an integer $k \geq 3$,
a satisfiable E$k$-CNF formula $\phi$, and
a pair of its satisfying assignments $\asg_\sss$ and $\asg_\ttt$,
there exists a polynomial-length reconfiguration sequence from $\asg_\sss$ to $\asg_\ttt$ in which
every assignment satisfies at least $\left(1-\frac{1}{k-1}-\frac{1}{k}\right)$-fraction of the clauses of $\phi$.
Moreover, such a reconfiguration sequence can be found by a deterministic polynomial-time algorithm.
In particular, this algorithm approximates \MMkSATReconf
within a factor of $1-\frac{1}{k-1}-\frac{1}{k}$.
\end{theorem}
\noindent
\cref{thm:intro:alg} implies a structural property of the solution space that
every pair of satisfying assignments for an E$k$-CNF formula
can be connected only by almost-satisfying assignments.
For small $k$, the proposed algorithm has an approximation factor much better than $1-\frac{1}{k-1}-\frac{1}{k}$,
as shown in \cref{tab:approx}.

\begin{table}[t]
    \centering
    \caption{Approximation factor of \MMkSATReconf for $3 \leq k \leq 10$.}
    \label{tab:approx}
    \small
    \begin{tabular}{ccccccccc}
    \toprule
        $k$ & $3$ & $4$ & $5$ & $6$ & $7$ & $8$ & $9$ & $10$ \\
    \midrule
        \textbf{approximation factor} & $0.572$ & $0.631$ & $0.679$ & $0.718$ & $0.749$ & $0.775$ & $0.796$ & $0.814$ \\
    \bottomrule
    \end{tabular}
\end{table}

On the hardness side, we show the $\PSPACE$-hardness of $\left(1 - \frac{1}{10k}\right)$-factor approximation
for every sufficiently large $k$.\footnote{
In \cref{sec:hard9.333}, we show
the $\PSPACE$-hardness of $\left(1-\frac{3-\epsilon}{28k}\right)$-factor approximation,
which is slightly better than $1 - \frac{1}{10k}$.
}

\begin{theorem}[informal; see \cref{thm:hard9.333}]
\label{thm:intro:hard9.333}
There exists an integer $k_0 \in \bbN$ such that
for any integer $k \geq k_0$,
a satisfiable E$k$-CNF formula $\phi$, and
a pair of its satisfying assignments $\asg_\sss$ and $\asg_\ttt$,
it is $\PSPACE$-hard to distinguish between the following two cases\textup{:}
\begin{itemize}
    \item \textup{(}Completeness\textup{)}
        There exists a reconfiguration sequence from $\asg_\sss$ to $\asg_\ttt$
        consisting of satisfying assignments for $\phi$.\footnote{
            This is a \Yes instance of \kSATReconf.
        }
    \item \textup{(}Soundness\textup{)}
        Every reconfiguration sequence  from $\asg_\sss$ to $\asg_\ttt$ contains an assignment that
        violates more than a $\frac{1}{10k}$-fraction of the clauses of $\phi$.
\end{itemize}
In particular, \MMkSATReconf is $\PSPACE$-hard to approximate within a factor of $1-\frac{1}{10k}$ for every integer $k \geq k_0$.
\end{theorem}
\noindent
We found this to be surprising.
For any E$k$-CNF formula $\phi$ over $n$ variables,
a random assignment $\asgrnd$ uniformly chosen from $\zo^n$ satisfies a $\left(1 - \frac{1}{2^k}\right)$-fraction of the clauses of $\phi$ in expectation.
By a concentration inequality,\footnote{
    We actually prove \cref{thm:intro:hard9.333} even for formulas $\phi$ such that each variable is read $o(|\phi|)$ times, and thus the read-$\tau$ concentration inequality is applicable.
}
this implies that only a $2^{-\Omega(n)}$-fraction of assignments do not satisfy a $\left(1 - \frac{1}{10k}\right)$-fraction of the clauses of $\phi$.
\cref{thm:intro:hard9.333} shows the $\PSPACE$-hardness of the $st$-connectivity question
over the subgraph of the $n$-dimensional Boolean hypercube
obtained by deleting only a $2^{-\Omega(n)}$-fraction of vertices.

As an immediate corollary of \cref{thm:intro:hard9.333},
we obtain the $\PSPACE$-hardness of $\left(1 - \Omega\left(\frac{1}{k}\right)\right)$-factor approximation for every $k \geq 3$.
\begin{corollary}[informal; see \cref{cor:hard0}]
\label{cor:intro:hard0}
There exists a universal constant $\delta_0 > 0$ such that
\MMkSATReconf is $\PSPACE$-hard to approximate
within a factor of $1-\frac{\delta_0}{k}$ for every integer $k \geq 3$.
\end{corollary}

\cref{thm:intro:alg,thm:intro:hard9.333} provide asymptotically tight lower and upper bounds
for approximability of \MMkSATReconf.
Note that the approximation threshold of its $\NP$ analogue,
i.e., \prb{Max E$k$-SAT}, is $1 - \frac{1}{2^k}$ \cite[Theorems~6.5 and~6.14]{hastad2001some}.
To the best of our knowledge,
this is the first reconfiguration problem whose approximation threshold is
(asymptotically) \emph{worse} than that of its $\NP$ analogue.

Prior to this work,
any reconfiguration problem has been shown to be at least as ``easy'' as its $\NP$-analogue in terms of approximability.
For example,
the approximation threshold of \prb{Minmax Set Cover Reconfiguration}\footnote{
In the \prb{Minmax Set Cover Reconfiguration} problem, 
we are asked to transform a given cover of a set system into another
by repeatedly adding or removing a single set so as to minimize
the maximum size of any covers during transformation.
} is $2$ 
\cite{ito2011complexity,karthikc.s.2023inapproximability,hirahara2024optimal} while
that of \prb{Min Set Cover} is $\ln N$
\cite{dinur2014analytical,feige1998threshold,johnson1974approximation,lovasz1975ratio,chvatal1979greedy},
where $N$ is the universe size.
See also \cref{tab:summary} and \cref{subsec:related:approx-reconf}
for the approximation threshold of other reconfiguration problems.
This trend comes from the nature of reconfiguration problems that
a pair of feasible solutions are given as input: it is often the case that
we can construct a trivial reconfiguration sequence that passes through an ``intermediate'' solution between them.
For example, for a pair of covers, their \emph{union} is also a cover at most twice as large,
which implies a $2$-factor approximation algorithm for \prb{Minmax Set Cover Reconfiguration} \cite[Theorem~6]{ito2011complexity}.
Contrary to this trend,
\MMkSATReconf exhibits a smaller approximation threshold than its $\NP$ analogue.
This indicates that the techniques from the PCP literature are not directly applicable to reconfiguration problems,
which hence suggests the need to develop new techniques.

\begin{table}[t]
\begin{threeparttable}
    \centering
    \caption{
        Approximation thresholds of reconfiguration problems and $\NP$ analogues.
        For the first three maximization problems, the larger the better.
        For the last minimization problem, the smaller the better.
    }
    \label{tab:summary}
    \small
    \begin{tabular}{l|l|l|l}
    \toprule
        \multicolumn{1}{c|}{\textbf{problem}} &
        \multicolumn{1}{c|}{\textbf{approx.~threshold}} &
        \multicolumn{1}{c|}{\textbf{hardness}} &
        \multicolumn{1}{c}{\textbf{refs.}} \\
    \midrule
        \prb{Maxmin E$k$-SAT Reconf} & $1-\Theta\left(\frac{1}{k}\right)$ & $\PSPACE$-h.
            & (this paper) \\
        \prb{Max E$k$-SAT} & $1-\frac{1}{2^k}$ & $\NP$-h.
            & \cite{hastad2001some} \\
    \midrule
        \prb{Maxmin $k$-Cut Reconf} & $1-\Theta\left(\frac{1}{k}\right)$ & $\PSPACE$-h.
            & \cite{hirahara2025asymptotically} \\
        \prb{Max $k$-Cut} & $1-\Theta\left(\frac{1}{k}\right)$ & $\NP$-h.
            & \cite{frieze1997improved,kann1997hardness,guruswami2013improved,austrin2014new} \\
    \midrule
        \prb{Maxmin 2-CSP Reconf} & $\Theta(1)$ & $\PSPACE$-h.
            & \cite{karthikc.s.2023inapproximability,ohsaka2024gap,ohsaka2025approximate} \\
        \prb{Max 2-CSP} & $N^{-\frac{1}{3}}$ to $2^{-(\log N)^{1-o(1)}}$ \tnote{$\dagger$} 
            & $\NP$-h.
            & \cite{raz1998parallel,charikar2011improved} \\
    \midrule
        \prb{Minmax Set Cover Reconf} & $2$ & $\PSPACE$-h.
            & \cite{ito2011complexity,karthikc.s.2023inapproximability,hirahara2024optimal} \\
        \prb{Min Set Cover} & $\ln N$ \tnote{$\ddagger$}
            & $\NP$-h.
            & \cite{dinur2014analytical,feige1998threshold,johnson1974approximation,lovasz1975ratio,chvatal1979greedy} \\
    \bottomrule
    \end{tabular}
    \begin{tablenotes}
        \item[$\dagger$] $N$ is the size of an instance of \prb{2-CSP}, which is equal to
            the number of variables times the alphabet size.
        \item[$\ddagger$] $N$ is the universe size of an instance of \prb{Set Cover}.
    \end{tablenotes}
\end{threeparttable}
\end{table}

\subsection{Organization}
The rest of this paper is organized as follows.
In \cref{sec:overview},
    we present an overview of the proof of \cref{thm:intro:alg,thm:intro:hard9.333}.
In \cref{sec:related},
    we review related work on
    reconfiguration problems,
    relatives of \kSATReconf, and
    approximability of reconfiguration problems and \prb{Max $k$-SAT}.
In \cref{sec:pre},
    we formally define the \MMkSATReconf problem and 
    introduce the Probabilistically Checkable Reconfiguration Proof theorem 
    \cite{hirahara2024probabilistically,karthikc.s.2023inapproximability}.
In \cref{sec:alg},
    we develop a deterministic $\left(1-\frac{1}{k-1}-\frac{1}{k}\right)$-factor approximation algorithm for \MMkSATReconf.
In \cref{sec:hard9.333}, we prove
    the $\PSPACE$-hardness of $\left(1-\frac{1}{10k}\right)$-factor approximation for \MMkSATReconf.
In \cref{sec:NP},
    we present a complementary result that
    \MMkSATReconf is $\NP$-hard to approximate within a factor of $1-\frac{1}{8k}$.
Some technical proofs are deferred to \cref{app}.

\section{Proof Overview}
\label{sec:overview}

\subsection{Deterministic \texorpdfstring{$\left(1-\frac{1}{k-1}-\frac{1}{k}\right)$-factor}{\protect{(1-1/(k-1)-1/k)-factor}} Approximation Algorithm (\texorpdfstring{\cref{sec:alg}}{Section~\ref{sec:alg}})}
First, we give a highlight of the proof of \cref{thm:intro:alg}, i.e.,
a deterministic $\left(1-\frac{1}{k-1}-\frac{1}{k}\right)$-factor approximation algorithm for \MMkSATReconf.
Our algorithm uses a random reconfiguration sequence passing through a random assignment.
A similar strategy was used to approximate other reconfiguration problems, e.g.,
\cite{hirahara2025asymptotically,karthikc.s.2023inapproximability,ohsaka2025approximate}.
Let $\phi$ be a satisfiable E$k$-CNF formula consisting of $m$ clauses $C_1, \ldots, C_m$
over $n$ variables $x_1, \ldots, x_n$, and
$\asg_\sss, \asg_\ttt \colon \{x_1, \ldots, x_n\} \to \zo$ be a pair of its satisfying assignments.
Let $\asgrnd \colon \{x_1, \ldots, x_n\} \to \zo$ be a random assignment for $\phi$,
which satisfies a $\left(1-\frac{1}{2^k}\right)$-fraction of the clauses of $\phi$ in expectation.
Consider the following two random reconfiguration sequences:
\begin{itemize}
    \item
        a reconfiguration sequence $\sq{\asg_1}$ from $\asg_\sss$ to $\asgrnd$
        obtained by flipping the assignment to variables at which $\asg_\sss$ and $\asgrnd$ differ in a random order, and
    \item
        a reconfiguration sequence $\sq{\asg_2}$ from $\asgrnd$ to $\asg_\ttt$ 
        obtained by flipping the assignment to variables at which $\asgrnd$ and $\asg_\ttt$ differ in a random order.
\end{itemize}
Concatenation of $\sq{\asg_1}$ and $\sq{\asg_2}$ yields
a reconfiguration sequence from $\asg_\sss$ to $\asg_\ttt$ that passes though $\asgrnd$,
which is obtained by the following procedure.

\begin{itembox}[l]{\textbf{Generating a random reconfiguration sequence $\sq{\asg_1} \circ \sq{\asg_2}$ from $\asg_\sss$ to $\asg_\ttt$}}
\begin{algorithmic}[1]
    \State sample a uniformly random assignment $\asgrnd \colon \{x_1, \ldots, x_n\} \to \zo$ for $\phi$.
    \LComment{start with $\asg_\sss$.}
    \For{\textbf{each} variable $x_i$ such that $\asg_\sss(x_i) \neq \asgrnd(x_i)$ in a random order}
        \State flip $x_i$'s current assignment from $\asg_\sss(x_i)$ to $\asgrnd(x_i)$.
    \EndFor
    \LComment{obtain $\asgrnd$.}
    \For{\textbf{each} variable $x_i$ such that $\asgrnd(x_i) \neq \asg_\ttt(x_i)$ in a random order}
        \State flip $x_i$'s current assignment from $\asgrnd(x_i)$ to $\asg_\ttt(x_i)$.
    \EndFor
    \LComment{end with $\asg_\ttt$.}
\end{algorithmic}
\end{itembox}

The main lemma is the following.

\begin{lemma}[informal; see \cref{lem:alg:round}]
\label{lem:intro:alg:round}
For each clause $C_j$ of $\phi$,
all assignments in $\sq{\asg_1} \circ \sq{\asg_2}$ simultaneously satisfy 
$C_j$ with probability at least $1-\frac{1}{k-1}-\frac{1}{k}$.
\end{lemma}
\noindent
The key insight in the proof of \cref{lem:intro:alg:round} is that
the probability of interest attains the minimum
when both $\asg_\sss$ and $\asg_\ttt$ make a single literal of $C_j$ true.
Thus, it is sufficient to bound from below the probability of interest only when
$\asg_\sss$ and $\asg_\ttt$ make a single literal of $C_j$ true and
($\asg_\sss \neq \asg_\ttt$ or $\asg_\sss = \asg_\ttt$),
which can be exactly calculated by exhaustion.
Derandomization can be done by a standard application of the method of conditional expectations \cite{alon2016probabilistic}.

\subsection{$\PSPACE$-hardness of \texorpdfstring{$\left(1-\frac{1}{10k}\right)$-factor}{(1-1/10k)-factor} Approximation (\texorpdfstring{\cref{sec:hard9.333}}{Section~\ref{sec:hard9.333}})}
Second, we present a proof overview of \cref{thm:intro:hard9.333}, i.e.,
$\PSPACE$-hardness of $\left(1-\frac{1}{10k}\right)$-factor approximation for \MMkSATReconf.
For a satisfiable E$k$-CNF formula $\phi$ and a pair of its satisfying assignments $\asg_\sss$ and $\asg_\ttt$,
let $\opt_\phi(\asg_\sss \reco \asg_\ttt)$ denote
the optimal value of \MMkSATReconf; namely,
the maximum value among all possible reconfiguration sequences from $\asg_\sss$ to $\asg_\ttt$,
where the \emph{value} of a reconfiguration sequence $\sq{\asg}$
is defined as the minimum fraction of satisfied clauses of $\phi$
over all assignments in $\sq{\asg}$.
For any reals $0 \leq s \leq c \leq 1$,
\prb{Gap$_{c,s}$ \kSATReconf} is a promise problem that requires to determine whether
$\opt_\phi(\asg_\sss \reco \asg_\ttt) \geq c$ or
$\opt_\phi(\asg_\sss \reco \asg_\ttt) < s$.
See \cref{subsec:pre:kSATReconf} for the formal definition.

\subsubsection{First Attempt: A Simple Proof of $\left(1-\Omega\left(\frac{1}{2^k}\right)\right)$-factor Inapproximability}
For starters,
we show the $\PSPACE$-hardness of $\left(1-\Omega\left(\frac{1}{2^k}\right)\right)$-factor approximation
for \MMkSATReconf.
The proof is based on
a simple gap-preserving reduction from \MMtreSATReconf to \MMkSATReconf,
which mimics that from \prb{Max E3-SAT} to \prb{Max E$k$-SAT}, e.g., \cite[Theorem~6.14]{hastad2001some}.
Let $\phi$ be a satisfiable E$3$-CNF formula over $n$ variables $x_1, \ldots, x_n$ and
$\asg_\sss, \asg_\ttt \colon \{x_1, \ldots, x_n\} \to \zo$ be a pair of its satisfying assignments.
Create fresh $K$ variables $y_1, \ldots, y_K$, where $K \defeq k-3$.
Construct an E$k$-CNF formula $\psi$
by appending the $2^K$ possible clauses over $y_1, \ldots, y_K$ to each clause of $\phi$.
Define two satisfying assignments
$\bsg_\sss,\bsg_\ttt \colon \{x_1, \ldots, x_n, y_1, \ldots, y_K\} \to \zo$ for $\psi$ such that
$\bsg_\sss|_{\{x_1, \ldots, x_n\}} \defeq \asg_\sss$,
$\bsg_\sss|_{\{y_1, \ldots, y_K\}} \defeq 0^K$,
$\bsg_\ttt|_{\{x_1, \ldots, x_n\}} \defeq \asg_\ttt$, and
$\bsg_\ttt|_{\{y_1, \ldots, y_K\}} \defeq 0^K$,
which completes the description of the reduction.
Observe easily that the following completeness and soundness hold:
\begin{itemize}
    \item (Completeness)
        If $\opt_\phi\bigl(\asg_\sss \reco \asg_\ttt\bigr) = 1$,
        then $\opt_\psi\bigl(\bsg_\sss \reco \bsg_\ttt\bigr) = 1$.
    \item (Soundness)
        If $\opt_\phi\bigl(\asg_\sss \reco \asg_\ttt\bigr) < 1-\epsilon$,
        then $\opt_\psi\bigl(\bsg_\sss \reco \bsg_\ttt\bigr) < 1-\frac{\epsilon}{2^{k-3}}$.
\end{itemize}
Since \prb{Gap$_{1,1-\epsilon}$ \treSATReconf} is $\PSPACE$-hard for some real $\epsilon > 0$
\cite{hirahara2024probabilistically,karthikc.s.2023inapproximability,ohsaka2023gap},
so is \prb{Gap$_{1,1-\frac{\epsilon}{2^{k-3}}}$ \kSATReconf}.
In particular,
\MMkSATReconf is $\PSPACE$-hard to approximate within a factor of $1 - \Omega\left(\frac{1}{2^k}\right)$.
To improve the inapproximability factor to $1 - \Omega\left(\frac{1}{k}\right)$,
we need to exploit some property that is possessed by \MMkSATReconf but not by \prb{Max E$k$-SAT}.
We achieve this by using a ``non-monotone'' test described next.

\subsubsection{The Power of Non-monotone Tests: $\left(1-\Omega\left(\frac{1}{1.913^k}\right)\right)$-factor Inapproximability}
We introduce the ``non-monotone'' test to prove
the $\PSPACE$-hardness of $1-\Omega\left(\frac{1}{1.913^k}\right)$-factor approximation for \MMkSATReconf
(for every $k$ divisible by $3$).
Let $\phi$ be a satisfiable E$3$-CNF formula consisting of $m$ clauses $C_1,\ldots, C_m$ over
$n$ variables $x_1, \ldots, x_n$ and
$\asg_\sss,\asg_\ttt \colon \{x_1, \ldots, x_n\} \to \zo$ be a pair of its satisfying assignments.
Let $\rep \geq 2$ be an integer and $k \defeq 3\rep$.
The \emph{Horn verifier} $\Vhorn$,
given oracle access to an assignment $\asg \colon \{x_1, \ldots, x_n\} \to \zo$,
selects $\rep$ clauses of $\phi$ randomly,
denoted by $C_{i_1}, \ldots, C_{i_{\rep}}$, and
accepts if the Horn-like condition
$C_{i_1} \vee \bar{C_{i_2}} \vee \cdots \vee \bar{C_{i_{\rep}}}$
is satisfied by $\asg$, as described below.

\begin{itembox}[l]{\textbf{$3\rep$-query Horn verifier $\Vhorn$ for an E$3$-CNF formula $\phi$}}
\begin{algorithmic}[1]
    \item[\textbf{Input:}]
        an E$3$-CNF formula $\phi = C_1 \wedge \cdots \wedge C_m$ over $n$ variables $x_1, \ldots, x_n$ and
        an integer $\rep \geq 2$.
    \item[\textbf{Oracle access:}]
        an assignment $\asg \colon \{x_1, \ldots, x_n\} \to \zo$.
    \State sample $i_1,\ldots,i_{\rep} \sim [m]$ uniformly at random.
    \If{$C_{i_1} \vee \bar{C_{i_2}} \vee \cdots \vee \bar{C_{i_{\rep}}}$ is satisfied by $\asg$}
        \Comment{(at most) $3\rep$ locations of $\asg$ are queried.}
        \State \Return $1$.
    \Else
        \State \Return $0$.
    \EndIf
\end{algorithmic}
\end{itembox}

\noindent
Intuitively, $\Vhorn$ thinks of each clause of $\phi$ as a new variable and
creates a kind of Horn clause on the fly.
If $\asg$ violates exactly $\epsilon$-fraction of the clauses of $\phi$,
then $\Vhorn$ rejects with probability $\epsilon(1-\epsilon)^{\rep-1}$,
which is ``non-monotone'' in $\epsilon$ and
attains the maximum at $\epsilon = \frac{1}{\rep}$ (see also \cref{fig:Horn}).
Let $\opt_{\Vhorn}(\asg_\sss \reco \asg_\ttt)$ denote the maximum value
among all possible reconfiguration sequences from $\asg_\sss$ to $\asg_\ttt$,
where the \emph{value} of a reconfiguration sequence $\sq{\asg}$
is defined as $\Vhorn$'s minimum acceptance probability over all assignments in $\sq{\asg}$.
The Horn verifier $\Vhorn$ has the following completeness and soundness:

\begin{figure}[t]
    \centering
    \includegraphics[width=0.55\linewidth]{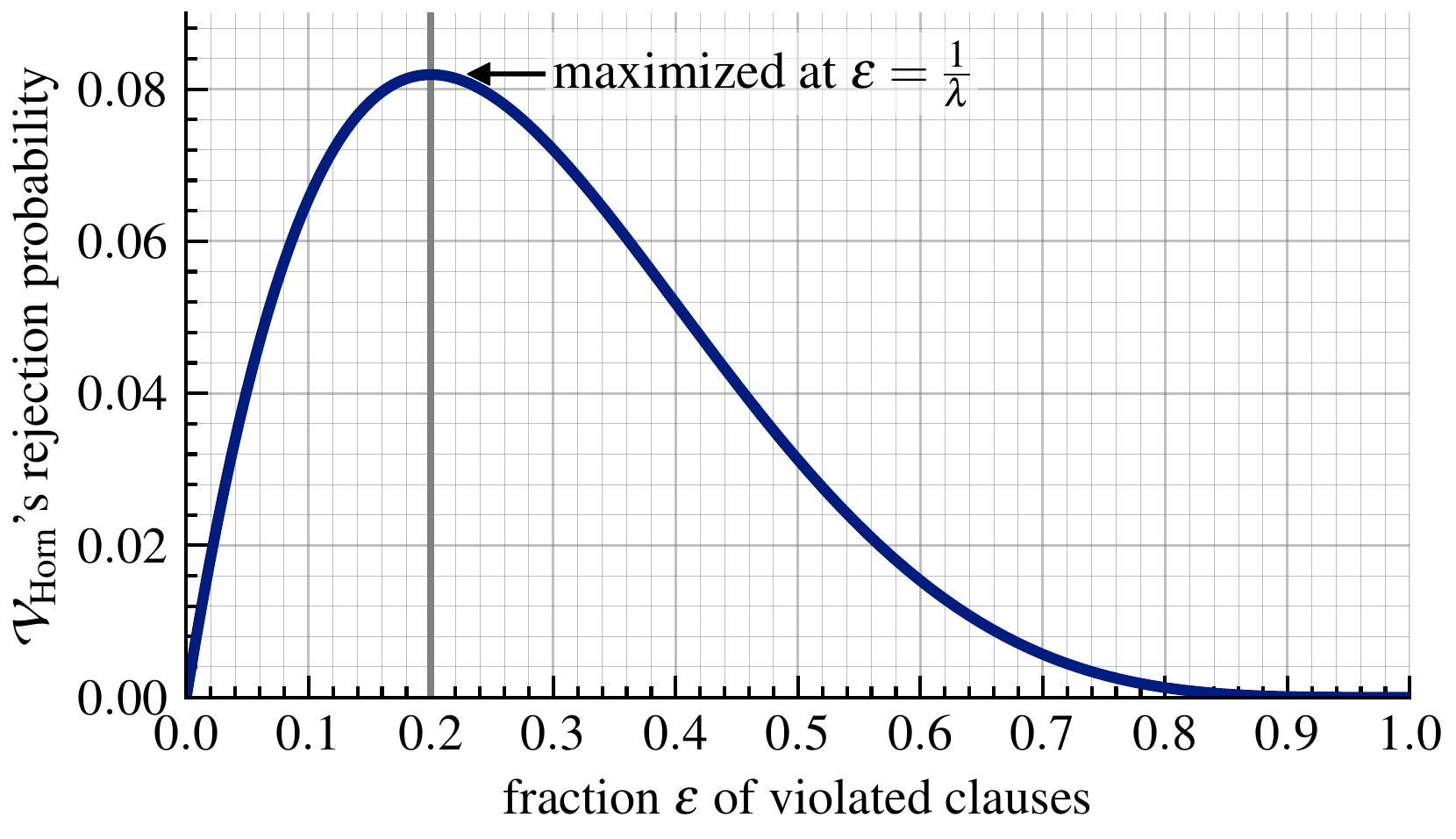}
    \caption{
        The rejection probability $\epsilon (1-\epsilon)^{\rep-1}$ of $\Vhorn$ parameterized by
        the fraction $\epsilon$ of violated clauses of an E$3$-CNF formula $\phi$
        (when $\rep = 5$).
        Obviously, $\epsilon (1-\epsilon)^{\rep-1}$ is not monotone in $\epsilon$ and
        attains the maximum at $\epsilon = \frac{1}{\rep}$.
        On the other hand,
        if an assignment violates most of the clauses of $\phi$ (i.e., $\epsilon \approx 1$),
        then $\Vhorn$ rejects it with only a tiny probability.
    }
    \label{fig:Horn}
\end{figure}

\begin{itemize}
    \item (Completeness)
        If $\opt_\phi\bigl(\asg_\sss \reco \asg_\ttt\bigr) = 1$,
        then $\opt_{\Vhorn}\bigl(\asg_\sss \reco \asg_\ttt\bigr) = 1$.
        This is immediate from the acceptance condition of $\Vhorn$.
    \item (Soundness)
        If $\opt_\phi\bigl(\asg_\sss \reco \asg_\ttt\bigr) < 1-\epsilon$,
        then $\opt_{\Vhorn}\bigl(\asg_\sss \reco \asg_\ttt\bigr) < 1 - \Omega\left(\frac{\epsilon}{\rep}\right)$.
        To see why this is true,
        let $\sq{\asg}$ be any reconfiguration sequence from $\asg_\sss$ to $\asg_\ttt$.
        By the soundness assumption,
        in order to transform $\asg_\sss$ into $\asg_\ttt$,
        we must violate more than $\epsilon$-fraction of the clauses of $\phi$ at some point.
        With this fact, we can show that
        $\sq{\asg}$ must contain some assignment $\asg^\circ$ that violates
        $\approx \frac{\epsilon}{\rep}$-fraction of the clauses of $\phi$.\footnote{
            In fact, we use \cite[Theorem~3.1]{ohsaka2023gap} to ensure that each variable of $\phi$ appears in a constant number of the clauses.
        }
        Such an assignment $\asg^\circ$ would be rejected by $\Vhorn$ with probability
        \begin{align}
            \Omega\left(\tfrac{\epsilon}{\rep} \cdot \left(1-\tfrac{\epsilon}{\rep}\right)^{\rep-1}\right)
            = \Omega\left(\frac{\epsilon}{\rep}\right).
        \end{align}
        See also \cref{fig:illust} for illustration.
\end{itemize}

\begin{figure}[t]
    \centering
    \includegraphics[width=0.55\linewidth]{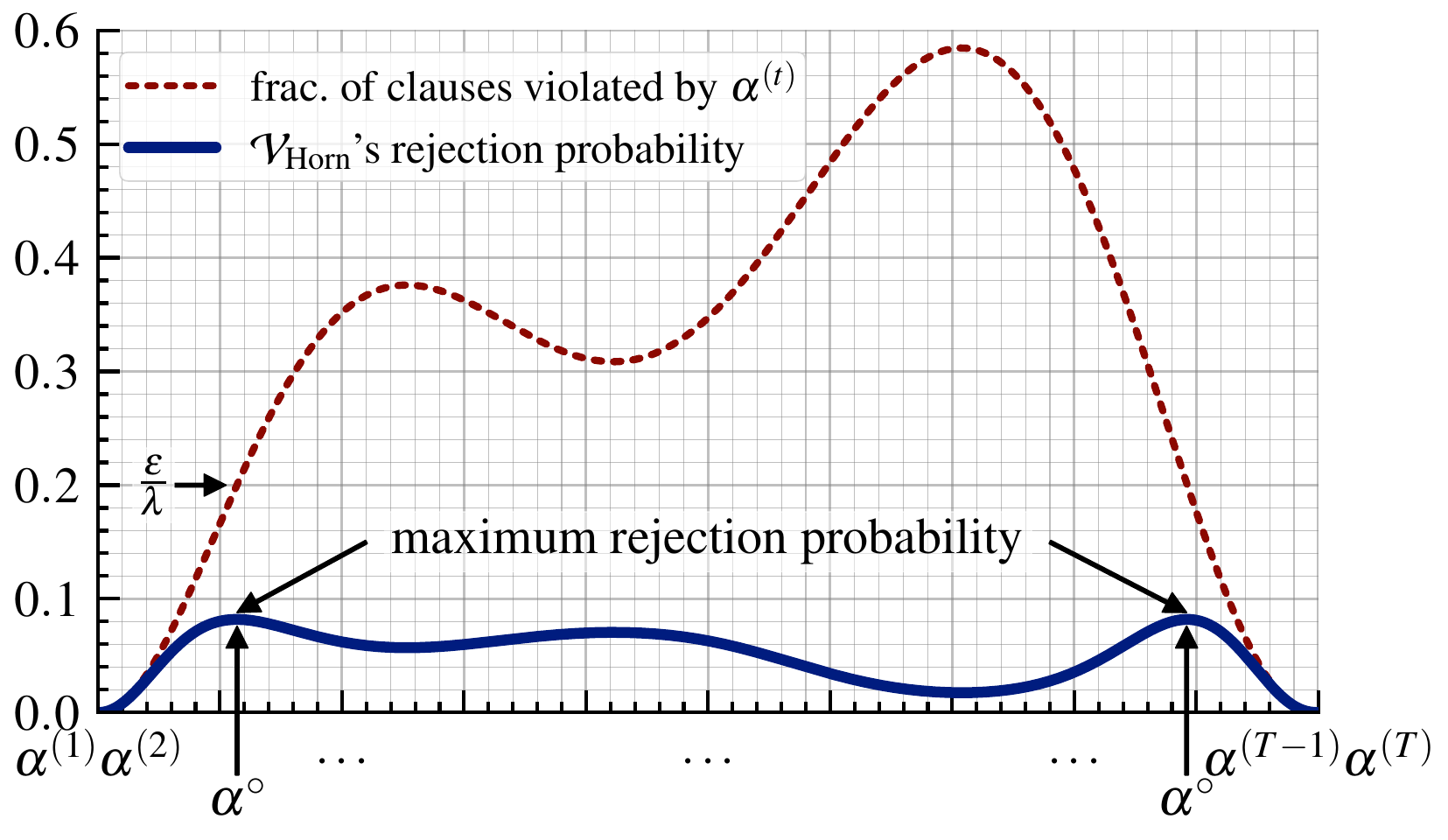}
    \caption{
        An example of the transition of
        the fraction of violated clauses and
        $\Vhorn$'s rejection probability (when $\rep = 5$).
        Let $\phi$ be a satisfiable E$3$-CNF formula,
        $\asg_\sss$ and $\asg_\ttt$ be a pair of its satisfying assignments, and
        $(\asg^{(1)}, \ldots, \asg^{(T)})$ be a reconfiguration sequence from $\asg_\sss$ to $\asg_\ttt$.
        The dotted (red) line represents the fraction of clauses of $\phi$ violated by $\asg^{(t)}$, and
        the solid (blue) line represents the probability that $\Vhorn$ rejects $\asg^{(t)}$.
        If $\opt_\phi(\asg_\sss \reco \asg_\ttt) < 1-\epsilon$,
        any reconfiguration sequence must contain some assignment $\asg^\circ$
        that violates $\approx \frac{\epsilon}{\rep}$-fraction of clauses of $\phi$,
        which would be rejected by $\Vhorn$ with probability $\Omega\left(\frac{\epsilon}{\rep}\right)$.
    }
    \label{fig:illust}
\end{figure}

Subsequently, we represent $\Vhorn$ by an E$k$-CNF formula.
For this purpose,
it is sufficient to ``emulate'' $\Vhorn$ by an \emph{\scOR-predicate verifier} $\X$,
which is allowed to generate a query sequence $I$ and a partial assignment $\tilde{\asg} \in \zo^I$, and
accepts if the local view $\asg|_I$ is \emph{not} equal to $\tilde{\asg}$.
The acceptance condition of $\X$ is equivalent to the following \scOR predicate:
$\bigvee_{i \in I} (\asg(i) \neq \tilde{\asg}(i))$.
Recall that $\Vhorn$ rejects if 
the Horn-like condition $C_{i_1} \vee \bar{C_{i_2}} \vee \cdots \vee \bar{C_{i_{\rep}}}$
is unsatisfied by $\asg$; namely,
its \emph{negation} is satisfied:
\begin{align}
\label{eq:intro:Horn-neg}
    \bar{C_{i_1}} \wedge C_{i_2} \wedge \cdots \wedge C_{i_{\rep}}.
\end{align}
There are $7^{\rep-1}$ possible (partial) assignments over $\zo^I$
that satisfy \cref{eq:intro:Horn-neg},
where $I$ is the set of variables appearing in $C_{i_1}, \ldots, C_{i_{\rep}}$.\footnote{
    In order for $I$ to contain \emph{exactly} $3\rep$ variables,
    the selected $\rep$ clauses $C_{i_1}, \ldots, C_{i_{\rep}}$ should not share common variables.
    Such an undesirable event occurs with negligible probability.
}
Since $\X$ can reject only a \emph{single} local view at a time,
it samples a partial assignment $\tilde{\asg} \in \zo^I$ satisfying \cref{eq:intro:Horn-neg}
uniformly at random and rejects if $\asg|_I = \tilde{\asg}$,
as described below.

\begin{itembox}[l]{\textbf{$3\rep$-query \scOR-predicate verifier $\X$ emulating $\Vhorn$}}
\begin{algorithmic}[1]
    \item[\textbf{Input:}]
        an E$3$-CNF formula $\phi = C_1 \wedge \cdots \wedge C_m$ over $n$ variables $x_1, \ldots, x_n$ and
        an integer $\rep \geq 2$.
    \item[\textbf{Oracle access:}]
        an assignment $\asg \colon \{x_1, \ldots, x_n\} \to \zo$.
    \State sample $i_1,\ldots,i_{\rep} \sim [m]$ uniformly at random.
    \State let $I$ be the set of variables appearing in $C_{i_1}, \ldots, C_{i_{\rep}}$.
    \State sample a partial assignment $\tilde{\asg} \in \zo^I$
    that satisfies \cref{eq:intro:Horn-neg} uniformly at random.
    \If{$\asg|_I \neq \tilde{\asg}$}
        \State \Return $1$.
    \Else
        \State \Return $0$.
    \EndIf
\end{algorithmic}
\end{itembox}

\noindent
The \scOR-predicate verifier $\X$ has the following completeness and soundness:
\begin{itemize}
    \item
        (Completeness)
        If $\opt_\phi\bigl(\asg_\sss \reco \asg_\ttt\bigr) = 1$, then
        $\opt_\X\bigl(\asg_\sss \reco \asg_\ttt\bigr) = 1$.
        This is immediate from the definition of $\X$.
    \item (Soundness)
        If $\opt_\phi\bigl(\asg_\sss \reco \asg_\ttt\bigr) < 1-\epsilon$, then
        $\opt_\X\bigl(\asg_\sss \reco \asg_\ttt\bigr) < 1 - \Omega\left(\frac{\epsilon}{1.913^k}\right)$.
        To see why this is true,
        let $\sq{\asg}$ be a reconfiguration sequence from $\asg_\sss$ to $\asg_\ttt$.
        By the soundness property of $\Vhorn$,
        there must be some assignment $\asg^\circ$ in $\sq{\asg}$ that is rejected by $\Vhorn$
        with probability $\Omega\left(\frac{\epsilon}{\rep}\right)$.
        Suppose that $\Vhorn$ rejects $\asg^\circ$ when examining the condition
        $C_{i_1} \vee \bar{C_{i_2}} \vee \cdots \vee \bar{C_{i_{\rep}}}$.
        Conditioned on this event, 
        we find $\X$ to reject $\asg^\circ$ with probability $\frac{1}{7^{\rep-1}}$
        since there are $7^{\rep-1}$ partial assignments that satisfy \cref{eq:intro:Horn-neg}.
        Therefore, the overall rejection probability of $\X$ is
        \begin{align}
        \label{eq:overview:PSPACE:emulate}
            \Pr\bigl[ \X \text{ rejects } \asg^\circ \bigr]
            = \underbrace{\Pr\bigl[ \Vhorn \text{ rejects } \asg^\circ \bigr]}_{= \Omega\left(\frac{\epsilon}{\rep}\right)}
                \cdot \frac{1}{7^{\rep-1}}
            = \Omega\left(\frac{\epsilon}{k \cdot 7^{\frac{1}{3}k}}\right)
            \underbrace{=}_{7^{\frac{1}{3}} < 1.913} \Omega\left(\frac{\epsilon}{1.913^k}\right).
        \end{align}
\end{itemize}

Consequently,
\prb{Gap$_{1,1-\epsilon}$ \treSATReconf} is reduced to
\prb{Gap$_{1,1-\Omega\left(\frac{\epsilon}{1.913^k}\right)}$ \kSATReconf}
for any real $\epsilon > 0$.
In particular,
\MMkSATReconf is $\PSPACE$-hard to approximate within a factor of
$1 - \Omega\left(\frac{1}{1.913^k}\right)$,
which is an exponential improvement over 
$1-\Omega\left(\frac{1}{2^k}\right)$.

\subsubsection{Getting $\left(1-\Omega\left(\frac{1}{k}\right)\right)$-factor Inapproximability}
To further reduce the inapproximability factor to $1-\Omega\left(\frac{1}{k}\right)$
as claimed in \cref{thm:intro:hard9.333},
we need to get rid of the $7^{\rep-1}$-factor appearing in \cref{eq:overview:PSPACE:emulate},
which is the number of partial assignments that satisfy \cref{eq:intro:Horn-neg},
i.e., $\bar{C_{i_1}} \wedge C_{i_2} \wedge \cdots \wedge C_{i_{\rep}}$.
For this purpose,
we shall replace each of $C_{i_2}, \ldots, C_{i_{\rep}}$ by
a DNF \emph{term} in the form of $\ell_1 \wedge \ell_2 \wedge \ell_3$
instead of a CNF \emph{clause} in the form of $\ell_1 \vee \ell_2 \vee \ell_3$,
so that the number of partial assignments
is reduced from $7^{\rep-1}$ to $\bigO(1)$, implying that
for any assignment $\asg \colon \{x_1, \ldots, x_n\} \to \zo$,
\begin{align}
\label{eq:intro:X-Vhorn}
    \Pr\bigl[ \X \text{ rejects } \asg \bigr]
    = \Omega\Bigl(\Pr\bigl[ \Vhorn \text{ rejects } \asg \bigr]\Bigr).
\end{align}
If this is the case,
$\opt_\phi(\asg_\sss \reco \asg_\ttt) < 1-\epsilon$
implies
$\opt_\X(\asg_\sss \reco \asg_\ttt) < 1 - \Omega\left(\frac{\epsilon}{k}\right)$.
We achieve this improvement by redesigning the Horn verifier $\Vhorn$ so as to
execute a PCRP system for \prb{Gap$_{1,1-\epsilon}$ \treSATReconf} and
a \emph{dummy verifier} $\A$, which accepts only a single prescribed string, say $1^n$,
with a carefully chosen probability.
Specifically, we develop the following three verifiers
(see \cref{sec:hard9.333} for the details):

\begin{itemize}
    \item
    The first verifier is the $3$-query \emph{combined verifier} $\W$.
    Given oracle access to a pair of 
    an assignment $\asg \colon \{x_1, \ldots, x_n\} \to \zo$ for $\phi$ and
    a proof $\qf \in \zo^n$,
    $\W$ performs the following:
    (1) with probability $\Theta\left(\frac{1}{k}\right)$,
    it selects a clause $C_i$ of $\phi$ randomly and 
    accepts if $C_i$ is satisfied by $\asg$, and
    (2) with probability $1-\Theta\left(\frac{1}{k}\right)$,
    it runs the dummy verifier $\A$ on $\qf$.
    The two proofs are defined as 
    $\Pf_\sss \defeq \asg_\sss \circ 1^n$ and
    $\Pf_\ttt \defeq \asg_\ttt \circ 1^n$.
    Observe easily that
    if $\opt_\phi(\asg_\sss \reco \asg_\ttt) = 1$,
    then $\opt_{\W}(\Pf_\sss \reco \Pf_\ttt) = 1$, and
    if $\opt_\phi(\asg_\sss \reco \asg_\ttt) < 1-\epsilon$,
    then $\opt_{\W}(\Pf_\sss \reco \Pf_\ttt) < 1-\Omega\left(\frac{\epsilon}{k}\right)$.

    \item
    The second verifier is the (modified) $k$-query \emph{Horn verifier} $\Vhorn$,
    which independently runs $\W$ once and
    runs $\A$ $\rep-1$ times.
    Then, $\Vhorn$ accepts if $\W$ accepts or any of the $\rep-1$ runs of $\A$ \emph{rejects}.
    Similarly to the discussion in the previous section, we can show that 
    if $\opt_\phi(\asg_\sss \reco \asg_\ttt) = 1$,
    then $\opt_{\Vhorn}(\Pf_\sss \reco \Pf_\ttt) = 1$, and
    if $\opt_\phi(\asg_\sss \reco \asg_\ttt) < 1-\epsilon$,
    then $\opt_{\Vhorn}(\Pf_\sss \reco \Pf_\ttt) < 1-\Omega\left(\frac{\epsilon}{k}\right)$.
    Note that the number of \emph{rejecting} local views of $\Vhorn$ is $\bigO(1)$.

    \item
    The final verifier is the (modified) $k$-query \emph{\scOR-predicate verifier} $\X$,
    which is used to ``emulate'' $\Vhorn$ as in the previous section.
    Owing to the changes made to $\Vhorn$,
    there is a linear relation between the rejection probabilities of $\X$ and $\Vhorn$ similar to \cref{eq:intro:X-Vhorn},
    implying that 
    \prb{Gap$_{1,1-\epsilon}$ \treSATReconf} can be reduced to
    \prb{Gap$_{1,1-\Omega\left(\frac{\epsilon}{k}\right)}$ \kSATReconf}
    for any real $\epsilon > 0$.
\end{itemize}

\subsubsection{Perspective and Open Problem}
In this study, we found that
a reconfiguration problem may have a worse approximation threshold than its $\NP$ analogue.
In the hardness proof, we developed the Horn verifier to
exemplify the usefulness of its ``non-monotone'' behavior.
Here, we clarify what monotone and non-monotone verifiers are and
why the non-monotonicity can be useful in the reconfiguration regime.
Suppose that there are two verifiers $\V$ and $\W$,
which have oracle access to the same proof $\pf \in \zo^n$.
For example, $\V$ is a $3$-query verifier for \MMtreSATReconf and
$\W$ is the $3\rep$-query Horn verifier,
as we saw in the previous sections.
Suppose also that
$\W$'s rejection probability is bounded from below by
the value of some function $f \colon [0,1] \to [0,1]$ evaluated at $\V$'s rejection probability; namely,
\begin{align}
    \forall \pf \in \zo^n,\quad
    \Pr\bigl[ \W \text{ rejects } \pf \bigr] \geq f\Bigl(\Pr\bigl[ \V \text{ rejects } \pf \bigr]\Bigr).
\end{align}
For example, $f(\epsilon) = \epsilon(1-\epsilon)^{\rep-1}$ in the case of the $3\rep$-query Horn verifier.
We say that $\W$ is \emph{monotone} if $f$ is monotonically increasing.
In the PCP regime, the soundness property typically requires the following condition:
\begin{align}
\begin{aligned}
    & \forall \pf \in \zo^n,\quad
    \Pr\bigl[ \V \text{ rejects } \pf \bigr] \geq \epsilon, \\
    \implies 
    & \forall \pf \in \zo^n,\quad
    \Pr\bigl[ \W \text{ rejects } \pf \bigr] \geq f(\epsilon).
\end{aligned}
\end{align}
Since we are concerned with bounding from below the \emph{minimum} rejection probability of the verifier,
$\W$ should be monotone in general; i.e.,
the non-monotonicity is not helpful in showing the (better) soundness.

By contrast, in the reconfiguration regime, $\W$ does \emph{not} need to be monotone in deriving the soundness.
Suppose that
every reconfiguration sequence $\sq{\pf}$ from $\pf_\sss$ to $\pf_\ttt$
contains a proof $\pf^\circ$ that is rejected by $\V$ with probability (approximately) $\epsilon$.
Then, regardless of whether $\W$ is monotone or not,
for every reconfiguration sequence $\sq{\pf}$ from $\pf_\sss$ to $\pf_\ttt$,
the \emph{maximum} rejection probability of $\W$ over all proofs in $\sq{\pf}$
is (approximately) greater than $f(\epsilon)$; namely,
\begin{align}
\begin{aligned}
    & \forall \sq{\pf} = (\pf_\sss, \ldots, \pf_\ttt),\quad
    \exists \pf^\circ \in \sq{\pf},\quad
    \Pr\bigl[ \V \text{ rejects } \pf^\circ \bigr] \approx \epsilon, \\
    \implies
    & \forall \sq{\pf} = (\pf_\sss, \ldots, \pf_\ttt),\quad
    \max_{\text{all } \pf^\circ \in \sq{\pf}}
    \Bigl\{ \Pr\bigl[ \W \text{ rejects } \pf^\circ \bigr] \Bigr\} \gtrapprox f(\epsilon).
\end{aligned}
\end{align}
As a result, there are more possible choices for the verifier $\W$ that can be used in the reduction.

We believe that the concept of non-monotone verifiers will find further applications
in $\PSPACE$-hardness of approximation for reconfiguration problems other than \MMkSATReconf.
An immediate open problem is to elucidate for which $\NP$ problem
its reconfiguration analogue becomes ``harder'' in terms of approximability.
Specifically, for what class of Boolean relations
does \prb{Maxmin Satisfiability Reconfiguration}
have a worse approximation threshold than \prb{Max Satisfiability}?

\section{Related Work}
\label{sec:related}

\subsection{Reconfiguration Problems}

In the field of \emph{combinatorial reconfiguration},
we study algorithmic problems and structural properties over the space of feasible solutions.
In the unified framework due to \citet{ito2011complexity},
a \emph{reconfiguration problem} is defined with respect to 
a combinatorial problem $\Pi$ called the \emph{source problem} and
a transformation rule $R$ over the feasible solutions of $\Pi$.
For an instance $\calI$ of $\Pi$ and 
a pair of its feasible solutions $S_\sss$ and $S_\ttt$,
the reconfiguration problem asks if $S_\sss$ can be transformed into $S_\ttt$
by repeatedly applying the transformation rule $R$ while
always preserving the feasibility of any intermediate solution.
Speaking differently, the reconfiguration problem concerns the $st$-connectivity over the \emph{configuration graph},
which is an (undirected) graph $G_{\calI,R}$ where
each node corresponds to a feasible solution of the given instance $\calI$ and
each link represents that its endpoints can be transformed into each other by applying $R$.
A pair of $S_\sss$ and $S_\ttt$ is a \Yes instance of the reconfiguration problem
if and only if there is an (undirected) path from $S_\sss$ to $S_\ttt$ on $G_{\calI,R}$.
Such a sequence of feasible solutions that forms a path on the configuration graph
is called a \emph{reconfiguration sequence}.
Reconfiguration problems may date back to motion planning \cite{hopcroft1984complexity} and classical puzzles, including 15 puzzles \cite{johnson1879notes} and Rubik's Cube.
Over the past two decades,
reconfiguration problems have been defined from many source problems.
For example,
reconfiguration problems of
\prb{3-SAT} \cite{gopalan2009connectivity},
\prb{4-Coloring} \cite{bonsma2009finding},
\prb{Independent Set} \cite{hearn2005pspace,hearn2009games,kaminski2012complexity}, and
\prb{Shortest Path} \cite{bonsma2013complexity}
are $\PSPACE$-complete, whereas
those of
\prb{2-SAT} \cite{gopalan2009connectivity},
\prb{3-Coloring} \cite{cereceda2011finding},
\prb{Matching} \cite{ito2011complexity}, and
\prb{Spanning Tree} \cite{ito2011complexity}
belong to $\cP$.
We refer the reader to the surveys by \citet{nishimura2018introduction,heuvel2013complexity,mynhardt2019reconfiguration,bousquet2024survey}
as well as
the Combinatorial Reconfiguration wiki \cite{hoang2024combinatorial} for
more algorithmic, hardness, and structural results of reconfiguration problems.

\subsection{Relatives of \kSATReconf}
\citet{gopalan2009connectivity}
initiated a systematic study on the reconfiguration problem of Boolean satisfiability.
By extending Schaefer's dichotomy theorem \cite{schaefer1978complexity},
which classifies the complexity of every \prb{Satisfiability} problem
as $\cP$ or $\NP$-complete,
\cite[Theorem 2.9]{gopalan2009connectivity} proved the following dichotomy theorem
for every \prb{Satisfiability Reconfiguration} problem:
the reconfiguration problem for Boolean formulas is
in $\cP$ if the formulas are built from \emph{tight relations} and 
is $\PSPACE$-complete otherwise.
Schaefer relations are tight but not vice versa, and thus,
the $\NP$-hardness of a particular \prb{Satisfiability} problem does \emph{not} necessarily imply
the $\PSPACE$-hardness of the corresponding \prb{Satisfiability Reconfiguration} problem; e.g.,
\prb{1-in-3 SAT Reconfiguration} is in $\cP$,
even though \prb{1-in-3 SAT} is $\NP$-complete.\footnote{
In the \prb{1-in-3 SAT} problem,
each clause of an input formula contains three literals, and
it is deemed satisfied if exactly one of the three literals is true.
}

Other than $st$-connectivity problems,
there are several types of reconfiguration problems         \cite{mouawad2015reconfiguration,nishimura2018introduction,heuvel2013complexity}.
One is \emph{connectivity problems} \cite{gopalan2009connectivity,makino2010boolean,makino2011exact},
which ask if the configuration graph is connected; i.e.,
every pair of satisfying assignments are reachable from each other.
There exists a trichotomy result that
determines whether the connectivity problem of \prb{Satisfiability}
is $\cP$, $\coNP$-complete, or $\PSPACE$-complete
\cite{gopalan2009connectivity,makino2010boolean,schwerdtfeger2012computational}.
Other algorithmic and structural problems related to \kSATReconf
include finding the shortest reconfiguration sequence \cite{mouawad2017shortest} and
investigating the diameter of the configuration graph \cite{gopalan2009connectivity}.

\subsection{Approximability of Reconfiguration Problems}
\label{subsec:related:approx-reconf}
For a reconfiguration problem,
its \emph{approximate version} allows to use infeasible solutions,
but requires to optimize the ``worst'' feasibility throughout the reconfiguration sequence.
In the language of configuration graphs, we would like to relax the feasibility until a given pair of feasible solutions become connected.
\citet[Theorems~4~and~5]{ito2011complexity}
showed that several reconfiguration problems are $\NP$-hard to approximate.
Since most reconfiguration problems are $\PSPACE$-complete \cite{nishimura2018introduction},
$\NP$-hardness results are not optimal.
The significance of showing $\PSPACE$-hardness compared to $\NP$-hardness is that
it disproves the existence of a witness (in particular, a reconfiguration sequence) of polynomial length assuming that $\NP \neq \PSPACE$, and
it rules out any polynomial-time algorithm under the weak assumption that $\cP \neq \PSPACE$.
\cite{ito2011complexity} posed
the $\PSPACE$-hardness of approximation for reconfiguration problems as an open problem.
\citet{ohsaka2023gap} postulated a reconfiguration analogue of the PCP theorem \cite{arora1998probabilistic,arora1998proof},
called the \emph{Reconfiguration Inapproximability Hypothesis} (RIH), and
proved that assuming RIH, approximate versions of several reconfiguration problems are $\PSPACE$-hard to approximate,
including those of 
\prb{3-SAT}, \prb{Independent Set}, \prb{Vertex Cover}, \prb{Clique}, and \prb{Set Cover}.
\citet[Theorem~1.5]{hirahara2024probabilistically} and
\citet[Theorem~1]{karthikc.s.2023inapproximability}
independently gave a proof of RIH by establishing the \emph{Probabilistically Checkable Reconfiguration Proof} (PCRP) theorem,
which provides a PCP-type characterization of $\PSPACE$.
The PCRP theorem, along with a series of gap-preserving reductions
\cite{ohsaka2023gap,ohsaka2024gap,ohsaka2024alphabet,hirahara2024probabilistically,hirahara2024optimal},
implies unconditional $\PSPACE$-hardness of approximation results for the reconfiguration problems listed above,
thereby resolving the open problem of \cite{ito2011complexity} affirmatively.

Since the PCRP theorem itself only implies $\PSPACE$-hardness of approximation within some constant factor,
\emph{explicit} factors of inapproximability have begun to be investigated for reconfiguration problems.
In the $\NP$ regime,
the \emph{parallel repetition theorem} of \citet{raz1998parallel} can be used to derive
many strong inapproximability results, e.g.,
\cite{hastad1999clique,hastad2001some,feige1998threshold,bellare1998free,zuckerman2007linear}.
However, for a reconfiguration analogue of two-prover games,
a naive parallel repetition does not reduce its soundness error \cite{ohsaka2025approximate}.
\citet{ohsaka2024gap} adapted Dinur's gap amplification \cite{dinur2007pcp,radhakrishnan2006gap,radhakrishnan2007dinurs} to
show that
\prb{Maxmin 2-CSP Reconfiguration} and \prb{Minmax Set Cover Reconfiguration} are 
$\PSPACE$-hard to approximate within a factor of $0.9942$ and $1.0029$, respectively.
\citet[Theorems~3 and~4]{karthikc.s.2023inapproximability} proved the $\NP$-hardness
of $\left(\frac{1}{2}+\epsilon\right)$-factor approximation for \prb{Maxmin 2-CSP Reconfiguration} and
of $(2-\epsilon)$-factor approximation for \prb{Minmax Set Cover Reconfiguration}
for any real $\epsilon > 0$.
These results are numerically tight because
\prb{Maxmin 2-CSP Reconfiguration} admits
    a $\left(\frac{1}{2} - \epsilon\right)$-factor approximation \cite[Theorem~6]{karthikc.s.2023inapproximability} and
\prb{Minmax Set Cover Reconfiguration} admits
    a $2$-factor approximation \cite[Theorem~6]{ito2011complexity}.
\citet{hirahara2024optimal} proved that \prb{Minmax Set Cover Reconfiguration} is
$\PSPACE$-hard to approximate within a factor of $2-o(1)$,
improving upon \cite{karthikc.s.2023inapproximability,ohsaka2024gap}.
This is the first optimal $\PSPACE$-hardness result for approximability of any reconfiguration problem.
\citet{hirahara2025asymptotically} showed that
the approximation threshold of \prb{Maxmin $k$-Cut Reconfiguration} lies in $1-\Theta\left(\frac{1}{k}\right)$.
Other reconfiguration problems for which approximation algorithms were developed include
\prb{Subset Sum Reconfiguration} \cite{ito2014approximability} and
\prb{Submodular Reconfiguration} \cite{ohsaka2022reconfiguration}.
\cref{tab:summary} summarizes existing approximation thresholds for reconfiguration problems
and their source problems.
Except for \MMkSATReconf,
every reconfiguration problem is at least as ``easy'' as its source problem in terms of approximability.

\subsection{Approximability of \prb{Max E$k$-SAT}}
The \prb{Max E$k$-SAT} problem
seeks an assignment for an E$k$-CNF formula that satisfies the maximum number of clauses.
Observe easily that a random assignment makes a $\left(1-\frac{1}{2^k}\right)$-fraction of 
clauses satisfied in expectation.
\citet[Theorems~6.5 and~6.14]{hastad1999clique} proved that this is tight; namely,
for every $k \geq 3$,
it is $\NP$-hard to approximate \prb{Max E$k$-SAT} within a factor of $1-\frac{1}{2^k}+\epsilon$
for any real $\epsilon > 0$.
For the special case of $k = 2$,
the best known approximation ratio of \prb{Max $2$-SAT} is 
$\beta_{\mathrm{LLZ}} \approx 0.940$ due to \citet{lewin2002improved}.
Under the Unique Games Conjecture \cite{khot2002power},
\prb{Max $2$-SAT} cannot be approximated in polynomial time within a factor of
$\beta_{\mathrm{LLZ}} + \epsilon$ for any real $\epsilon > 0$
\cite{austrin2007balanced,brakensiek2024tight}.

\section{Preliminaries}
\label{sec:pre}
Let $\bbN \defeq \{0,1,2,3,\ldots\}$ denote the set of all nonnegative integers.
For a nonnegative integer $n \in \bbN$, let $[n] \defeq \{1,2,3,\ldots,n\}$.
The base of logarithms is $2$.
For a (finite) set $S$ and a nonnegative integer $k \in \bbN$,
we write $\binom{S}{k}$ for the family of all size-$k$ subsets of $S$.
We use the Iverson bracket $\llbracket \cdot \rrbracket$; i.e.,
for a statement $P$,
we define $\llbracket P \rrbracket$ as $1$ if $P$ is true and $0$ otherwise.
A \emph{sequence} of a finite number of elements $a^{(1)}, \ldots, a^{(T)}$
is denoted by $\sq{a} = (a^{(1)}, \ldots, a^{(T)})$, and
we write $a \in \sq{a}$ to indicate that $a$ appears in $\sq{a}$ (at least once).
The symbol $\circ$ stands for a concatenation of two sequences or functions.
For a set $S$, we write $X \sim S$ to mean that 
$X$ is a random variable uniformly drawn from $S$.
For a function $f \colon D \to R$ over a finite domain $D$ and its subset $I \subset D$,
we use $f|_I \colon I \to R$ to denote the \emph{restriction} of $f$ to $I$.
We write
$0^n$ for $\underbrace{0 \cdots 0}_{n \text{ times}}$ and
$1^n$ for $\underbrace{1 \cdots 1}_{n \text{ times}}$.

\subsection{Definition of \MMkSATReconf}
\label{subsec:pre:kSATReconf}
We define \kSATReconf and its approximate version.
We use the standard terminology and notation of Boolean satisfiability.
A \emph{Boolean formula} $\phi$ consists of Boolean variables, denoted by $x_1, \ldots, x_n$, and
the logical operators, denoted by \scAND ($\wedge$), \scOR ($\vee$), and \scNOT ($\neg$).
An \emph{assignment} for Boolean formula $\phi$ is defined as
a mapping $\asg \colon \{x_1, \ldots, x_n\} \to \zo$ that assigns
a truth value of $\zo$ to each variable $x_i$ of $\phi$.
We say that $\asg$ \emph{satisfies} $\phi$ if $\phi$ evaluates to $1$ when each
variable $x_i$ is assigned the truth value specified by $\asg(x_i)$.
We say that $\phi$ is \emph{satisfiable} if 
there exists an assignment $\asg$ that satisfies $\phi$.
A \emph{literal} is either a variable $x_i$ or its negation $\bar{x_i}$, and
a \emph{clause} is a disjunction of literals.
A Boolean formula is in \emph{conjunctive normal form} (CNF) 
if it is a conjunction of clauses.
By abuse of notation,
for an assignment $\asg \colon \{x_1, \ldots, x_n\} \to \zo$,
we write $\asg(\bar{x_i}) \defeq \bar{\asg(x_i)}$ for a negative literal $\bar{x_i}$, and
write $\asg(\ell_1, \ldots, \ell_k) \defeq (\asg(\ell_1), \ldots, \asg(\ell_k))$ for $k$ literals $\ell_1, \ldots, \ell_k$.
The \emph{width} of a clause is defined as the number of literals in it.
A \emph{$k$-CNF formula} is a CNF formula of width at most $k$, and
an \emph{E$k$-CNF formula} is a CNF formula of which every clause has width exactly $k$.

For a CNF formula $\phi$ over $n$ variables $x_1, \ldots, x_n$ and
a pair of its assignments $\asg_\sss, \asg_\ttt \colon \{x_1, \ldots, x_n\} \to \zo$,
a \emph{reconfiguration sequence} from $\asg_\sss$ to $\asg_\ttt$
is defined as a sequence 
$\sq{\asg} = (\asg^{(1)}, \ldots, \asg^{(T)})$
over assignments for $\phi$ such that
$\asg^{(1)} = \asg_\sss$,
$\asg^{(T)} = \asg_\ttt$, 
and
every adjacent pair of assignments differ in at most one variable
(i.e., $\asg^{(t)}(x_i) = \asg^{(t+1)}(x_i)$ for all but at most one variable $x_i$).
We sometimes call $\asg_\sss$ and $\asg_\ttt$ the \emph{starting} and \emph{ending} assignments.
In the \kSATReconf problem \cite{gopalan2009connectivity},
for a satisfiable E$k$-CNF formula $\phi$ and
a pair of its satisfying assignments $\asg_\sss$ and $\asg_\ttt$,
we are asked to decide if there exists a reconfiguration sequence from $\asg_\sss$ to $\asg_\ttt$
consisting only of satisfying assignments for $\phi$.
Note that \kSATReconf is $\PSPACE$-complete for every $k \geq 3$ \cite{gopalan2009connectivity}.

We formulate an approximate version of \kSATReconf.
Let $\phi$ be a CNF formula consisting of $m$ clauses $C_1, \ldots, C_m$ over $n$ variables $x_1, \ldots, x_n$.
The \emph{value} of an assignment $\asg \colon \{x_1, \ldots, x_n\} \to \zo$ for $\phi$,
denoted by $\val_\phi(\asg)$,
is defined as the fraction of clauses of $\phi$ satisfied by $\asg$; namely,
\begin{align}
    \val_\phi(\asg) \defeq \frac{1}{m}\cdot
    \left|\bigl\{
        j \in [m] \bigm| \asg \text{ satisfies } C_j
    \bigr\}\right|
    = \Pr_{j \sim [m]}\bigl[
        \asg \text{ satisfies } C_j
    \bigr].
\end{align}
The \emph{value} of a reconfiguration sequence $\sq{\asg} = (\asg^{(1)}, \ldots, \asg^{(T)})$ for $\phi$,
denoted by $\val_\phi(\sq{\asg})$,
is defined as the minimum fraction of satisfied clauses of $\phi$
over all assignments in $\sq{\asg}$; namely,
\begin{align}
    \val_\phi(\sq{\asg}) \defeq 
    \min_{1 \leq t \leq T} \val_\phi\bigl(\asg^{(t)}\bigr).
\end{align}
The \MMkSATReconf problem is defined as follows.
\begin{problem}
For a satisfiable E$k$-CNF formula $\phi$ and a pair of its satisfying assignments $\asg_\sss$ and $\asg_\ttt$,
\MMkSATReconf requires
to find a reconfiguration sequence $\sq{\asg}$ from $\asg_\sss$ to $\asg_\ttt$ such that
$\val_\phi(\sq{\asg})$ is maximized.
\end{problem}
\noindent
Let $\opt_\phi(\asg_\sss \reco \asg_\ttt)$
denote the optimal value of \MMkSATReconf,
which is the maximum of $\val_\phi(\sq{\asg})$
over all possible reconfiguration sequences
$\sq{\asg}$ from $\asg_\sss$ to $\asg_\ttt$; namely,
\begin{align}
    \opt_\phi\bigl(\asg_\sss \reco \asg_\ttt\bigr)
    \defeq \max_{\sq{\asg} = (\asg_\sss, \ldots, \asg_\ttt)} \val_\phi(\sq{\asg}).
\end{align}
The gap version of \MMkSATReconf is defined as follows.
\begin{problem}
For any integer $k \in \bbN$ and
any reals $c$ and $s$ with $0 \leq s \leq c \leq 1$,
\prb{Gap$_{c,s}$ \kSATReconf} requires to determine
for a satisfiable E$k$-CNF formula $\phi$ and a pair of its satisfying assignments $\asg_\sss$ and $\asg_\ttt$,
whether $\opt_\phi(\asg_\sss \reco \asg_\ttt) \geq c$ or
$\opt_\phi(\asg_\sss \reco \asg_\ttt) < s$.
\end{problem}
\noindent
In particular, the case of $s=c=1$ reduces to \kSATReconf.

\subsection{Probabilistically Checkable Reconfiguration Proofs}
First, we formalize the notion of \emph{verifier}.

\begin{definition}
A \emph{verifier} with
\emph{randomness complexity} $r \colon \bbN \to \bbN$ and
\emph{query complexity} $q \colon \bbN \to \bbN$
is a probabilistic polynomial-time algorithm $\V$
that given an input $x \in \zo^*$,
draws $r \defeq r(|x|)$ random bits $R \in \zo^r$ and uses
$R$ to generate a sequence of $q \defeq q(|x|)$ queries
$I = (i_1, \ldots, i_q)$ and a circuit $D \colon \zo^q \to \zo$.
We write $(I,D) \sim \V(x)$ to denote the random variable
for a pair of the query sequence and circuit generated by $\V$ on
input $x \in \zo^*$ and $r$ random bits.
Given an input $x \in \zo^*$ and oracle access to a \emph{proof} $\pf \in \zo^*$,
we define $\V$'s (randomized) output as 
a random variable $\V^\pf(x) \coloneq D(\pf|_I)$ for $(I,D) \sim \V(x)$ over the randomness of $R$.
We say that $\V(x)$ \emph{accepts} $\pf$ or simply $\V^{\pf}(x)$ \emph{accepts} if 
$\V^{\pf}(x) = 1$, and that
$\V^{\pf}(x)$ \emph{rejects} if $\V^{\pf}(x) = 0$.
\end{definition}

Then, we introduce the \emph{Probabilistically Checkable Reconfiguration Proof} (PCRP) theorem due to 
\citet{hirahara2024probabilistically,karthikc.s.2023inapproximability},
which offers a PCP-type characterization of $\PSPACE$.
A \emph{PCRP system} is defined as a triplet of
a verifier $\V$ and polynomial-time computable proofs $\pf_\sss,\pf_\ttt \colon \zo^* \to \zo^*$.
For a pair of \emph{starting} and \emph{ending} proofs $\pf_\sss, \pf_\ttt \in \zo^\ell$,
a \emph{reconfiguration sequence} from $\pf_\sss$ to $\pf_\ttt$
is defined as a sequence
$(\pf^{(1)}, \ldots, \pf^{(T)})$ over $\zo^\ell$ such that
$\pf^{(1)} = \pf_\sss$,
$\pf^{(T)} = \pf_\ttt$, and 
$\pf^{(t)}$ and $\pf^{(t+1)}$ differ in at most one bit for every $t \in [T-1]$.

\begin{theorem}[Probabilistically Checkable Reconfiguration Proof theorem \cite{hirahara2024probabilistically,karthikc.s.2023inapproximability}]
\label{thm:PCRP}
A language $L \subseteq \zo^*$ is in $\PSPACE$
if and only if
there exists a verifier $\V$ with
randomness complexity $r(n) = \bigO(\log n)$ and
query complexity $q(n) = \bigO(1)$,
coupled with polynomial-time computable proofs
$\pf_\sss, \pf_\ttt \colon \zo^* \to \zo^*$, such that
the following hold for every input $x \in \zo^*$\textup{:}
\begin{itemize}
    \item \textup{(}Completeness\textup{)}
    If $x \in L$, then there exists a reconfiguration sequence
    $\sq{\pf} = ( \pf^{(1)}, \ldots, \pf^{(T)} )$
    from $\pf_\sss(x)$ to $\pf_\ttt(x)$ such that
    $\V(x)$ accepts every proof in $\sq{\pf}$ with probability $1$\textup{;} namely,
    \begin{align}
        \forall t \in [T], \quad \Pr\Bigl[ \V^{\pf^{(t)}}(x) = 1 \Bigr] = 1.
    \end{align}
    \item \textup{(}Soundness\textup{)}
    If $x \notin L$, then every reconfiguration sequence 
    $\sq{\pf} = ( \pf^{(1)}, \ldots, \pf^{(T)} )$
    from $\pf_\sss(x)$ to $\pf_\ttt(x)$
    contains some proof that is rejected by $\V(x)$ with probability more than $\frac{1}{2}$\textup{;} namely,
    \begin{align}
        \exists t \in [T], \quad \Pr\Bigl[ V^{\pf^{(t)}}(x) = 1 \Bigr] < \frac{1}{2}.
    \end{align}
\end{itemize}    
\end{theorem}

For a verifier $\V$ and a reconfiguration sequence $\sq{\pf} = (\pf^{(1)}, \ldots, \pf^{(T)})$,
let $\val_\V(\sq{\pf})$ denote the minimum acceptance probability of $\V$
over all proofs in $\sq{\pf}$; namely,
\begin{align}
    \val_\V(\sq{\pf})
    \defeq \min_{1 \leq t \leq T} \Pr\bigl[ \V \text{ accepts } \pf^{(t)} \bigr].
\end{align}
For a verifier $\V$ and a pair of proofs $\pf_\sss,\pf_\ttt \in \zo^*$,
let $\opt_\V(\pf_\sss \reco \pf_\ttt)$
denote the maximum of $\val_\V(\sq{\pf})$
over all possible reconfiguration sequences $\sq{\pf}$ from $\pf_\sss$ to $\pf_\ttt$; namely,
\begin{align}
    \opt_\V\bigl(\pf_\sss \reco \pf_\ttt\bigr)
    \defeq \max_{\sq{\pf} = (\pf_\sss, \ldots, \pf_\ttt)} \val_\V(\sq{\pf}).
\end{align}
We say that
a PCRP system $(\V,\pf_\sss,\pf_\ttt)$ for a language $L \subseteq \zo^*$
has \emph{completeness} $c \colon \bbN \to \bbN$ and \emph{soundness} $s \colon \bbN \to \bbN$ if
the following hold for every input $x \in \zo^*$:
\begin{itemize}
    \item If $x \in L$, then 
        $\opt_{\V(x)}\bigl(\pf_\sss(x) \reco \pf_\ttt(x)\bigr) \geq c(|x|)$.
    \item If $x \notin L$, then 
        $\opt_{\V(x)}\bigl(\pf_\sss(x) \reco \pf_\ttt(x)\bigr) < s(|x|)$.
\end{itemize}
Note that the PCRP system of \cref{thm:PCRP} has
perfect completeness $c(n) = 1$ and soundness $s(n) = \frac{1}{2}$.

For a verifier $\V$ with randomness complexity $r \colon \bbN \to \bbN$ and
an input $x \in \zo^*$,
the \emph{degree} of a proof location $i$ is defined as
the number of random bit strings $R \in \zo^{r(|x|)}$ on which $\V(x)$ queries $i$; namely,
\begin{align}
    \bigl|\bigl\{ R \in \zo^{r(|x|)} \bigm| i \in I_R \bigr\}\bigr|
    = \Pr_{(I,D) \sim \V(x)}\bigl[i \in I\bigr] \cdot 2^{r(|x|)},
\end{align}
where
$I_R$ is the query sequence generated by $\V(x)$ over randomness $R$.
We say that $\V$ has the \emph{degree} $\Delta \colon \bbN \to \bbN$ if 
for every input $x \in \zo^*$,
each proof location has degree at most $\Delta(|x|)$.

\section{Deterministic \texorpdfstring{$\left(1-\frac{1}{k-1}-\frac{1}{k}\right)$-factor}{(1-1/(k-1)-1/k)-factor} Approximation Algorithm for \MMkSATReconf}
\label{sec:alg}

In this section, we prove \cref{thm:intro:alg}; i.e.,
we develop a deterministic $\left(1-\frac{1}{k-1}-\frac{1}{k}\right)$-factor approximation algorithm
for \MMkSATReconf for every $k \geq 3$.

\begin{theorem}
\label{thm:alg}
For an integer $k \geq 3$,
a satisfiable E$k$-CNF formula $\phi$, and
a pair of its satisfying assignments $\asg_\sss$ and $\asg_\ttt$,
there exists a polynomial-length reconfiguration sequence $\sq{\asg}$ from $\asg_\sss$ to $\asg_\ttt$ such that
\begin{align}
\begin{aligned}
    \val_\phi(\sq{\asg})
    \geq 1 - \frac{1}{k-1} - \frac{1}{k}.
\end{aligned}
\end{align}
Moreover, such $\sq{\asg}$ can be found by a deterministic polynomial-time algorithm.
In particular, this is a deterministic $\left(1-\frac{1}{k-1}-\frac{1}{k}\right)$-factor
approximation algorithm for \MMkSATReconf.
\end{theorem}

Some definitions are further introduced.
Let $\phi$ be a CNF formula consisting of $m$ clauses $C_1, \ldots, C_m$ over $n$ variables $x_1, \ldots, x_n$.
For a reconfiguration sequence $\sq{\asg}$ over assignments for $\phi$,
we say that $\sq{\asg}$ \emph{satisfies} a clause $C_j$ of $\phi$
if every assignment in $\sq{\asg}$ satisfies $C_j$.
For two assignments $\asg, \bsg \colon \{x_1, \ldots, x_n\} \to \zo$ for $\phi$,
let $\asg \triangle \bsg$ denote
the set of variables at which $\asg$ and $\bsg$ differ; namely,
\begin{align}
    \asg \triangle \bsg \defeq \bigl\{ x_i \bigm| \asg(x_i) \neq \bsg(x_i) \bigr\}.
\end{align}
For a pair of assignments $\asg_\sss,\asg_\ttt \colon \{x_1, \ldots, x_n\} \to \zo$ for $\phi$,
we say that
a reconfiguration sequence $\sq{\asg} = (\asg^{(1)}, \ldots, \asg^{(T)})$ from $\asg_\sss$ to $\asg_\ttt$ is \emph{irredundant} if
no adjacent pair of assignments in $\sq{\asg}$ are identical, and
for each variable $x_i$,
there is an index $t_i \in [T]$ such that
\begin{align}
    \asg^{(t)}(x_i) = 
    \begin{cases}
        \asg_\sss(x_i) & \text{if } t \leq t_i, \\
        \asg_\ttt(x_i) & \text{if } t > t_i.
    \end{cases}
\end{align}
In other words,
$\sq{\asg}$ is obtained by flipping the assignments to variables of $\asg_\sss \triangle \asg_\ttt$ exactly once in some order.
For two assignments $\asg_1, \asg_2 \colon \{x_1, \ldots, x_n\} \to \zo$,
let $\stsqasg(\asg_1 \reco \asg_2)$ denote
the set of all irredundant reconfiguration sequences from $\asg_1$ to $\asg_2$.
For three assignments $\asg_1, \asg_2, \asg_3 \colon \{x_1, \ldots, x_n\} \to \zo$,
let $\stsqasg(\asg_1 \reco \asg_2 \reco \asg_3)$ denote
the set of reconfiguration sequences from $\asg_1$ to $\asg_3$ obtained by concatenating
all possible pairs of irredundant reconfiguration sequences of
$\stsqasg(\asg_1 \reco \asg_2)$ and $\stsqasg(\asg_2 \reco \asg_3)$.

The proof of \cref{thm:alg} relies on the following lemma,
which states that
a random reconfiguration sequence 
that passes through a random assignment satisfies each clause
with probability $1 - \frac{1}{k-1} - \frac{1}{k}$.

\begin{lemma}
\label{lem:alg:round}
Let $k \geq 3$ be an integer,
$x_1, \ldots, x_k$ be $k$ variables,
$C = \ell_1 \vee \cdots \vee \ell_k$ be a clause of width $k$ over $x_1, \ldots, x_k$, and
$\asg_\sss, \asg_\ttt \colon \{x_1, \ldots, x_k\} \to \zo$ be a pair of satisfying assignments for $C$.
Consider a uniformly random assignment $\asgrnd \colon \{x_1, \ldots, x_k\} \to \zo$ and
a random reconfiguration sequence $\sq{\asg}$ from $\asg_\sss$ to $\asg_\ttt$
uniformly chosen from $\stsqasg(\asg_\sss \reco \asgrnd \reco \asg_\ttt)$.
Then, $\sq{\asg}$ satisfies $C$
with probability at least $1 - \frac{1}{k-1} - \frac{1}{k}$\textup{;} namely,
\begin{align}
\label{eq:alg:round}
\begin{aligned}
    \Pr_{\asgrnd, \sq{\asg}}\bigl[
        \sq{\asg} \text{ satisfies } C
    \bigr]
    \geq 1 - \frac{1}{k-1} - \frac{1}{k}.
\end{aligned}
\end{align}
\end{lemma}

By using \cref{lem:alg:round}, we can prove \cref{thm:alg}.

\begin{proof}[Proof of \cref{thm:alg}]
Let $\phi$ be a satisfiable E$k$-CNF formula consisting of $m$ clauses
$C_1, \ldots, C_m$ over $n$ variables $x_1, \ldots, x_k$, and
$\asg_\sss,\asg_\ttt \colon \{x_1, \ldots, x_k\} \to \zo$ be 
a pair of satisfying assignments for $\phi$.
Let
$\asgrnd \colon \{x_1, \ldots, x_k\} \to \zo$ be a uniformly random assignment and
$\sq{\asg}$ be a random reconfiguration sequence uniformly chosen from
$\stsqasg(\asg_\sss \reco \asgrnd \reco \asg_\ttt)$.
By linearity of expectation and \cref{lem:alg:round},
we derive
\begin{align}
    \E_{\asgrnd, \sq{\asg}}\bigl[
        \val_\phi(\sq{\asg})
    \bigr]
    \geq \frac{1}{m} \cdot \sum_{1 \leq j \leq m} \Pr_{\asgrnd, \sq{\asg}}\bigl[
        \sq{\asg} \text{ satisfies } C_j
    \bigr]
    \geq 1-\frac{1}{k-1}-\frac{1}{k}.
\end{align}
By a standard application of the method of conditional expectations \cite{alon2016probabilistic},
we can construct a reconfiguration sequence $\sq{\asg}^*$ from $\asg_\sss$ to $\asg_\ttt$ such that
\begin{align}
    \val_\phi(\sq{\asg}^*) \geq 1-\frac{1}{k-1}-\frac{1}{k}
\end{align}
in deterministic polynomial time, which accomplishes the proof.
\end{proof}

The remainder of this section is devoted to the proof of \cref{lem:alg:round}.
Hereafter,
we fix the set of $k$ variables, denoted by $V \defeq \{x_1, \ldots, x_k\}$, and
fix a clause of width $k$, denoted by $C = \ell_1 \vee \cdots \vee \ell_k$.
For the sake of simplicity, we assume that each literal $\ell_i$ is either $x_i$ or $\bar{x_i}$.
We first show that the probability of interest---the left-hand side of \cref{eq:alg:round}---is \emph{monotone}
with respect to $\asg_\sss$ and $\asg_\ttt$.
\begin{claim}
\label{clm:alg:monotone}
Let $\asg_\sss, \asg_\ttt, \asg'_\sss, \asg'_\ttt \colon V \to \zo$ be
satisfying assignments for $C$.
Suppose that for each literal $\ell_i$ of $C$,
$\asg_\sss(\ell_i) = 1$ implies $\asg'_\sss(\ell_i) = 1$ and
$\asg_\ttt(\ell_i) = 1$ implies $\asg'_\ttt(\ell_i) = 1$\textup{;} namely,
$\asg_\sss(\ell_i) \leq \asg'_\sss(\ell_i)$ and $\asg_\ttt(\ell_i) \leq \asg'_\ttt(\ell_i)$.
For a uniformly random assignment
$\asgrnd \colon V \to \zo$ and
four random irredundant reconfiguration sequences
$\sq{\asg_1} \sim \stsqasg(\asg_\sss \reco \asgrnd)$,
$\sq{\asg_2} \sim \stsqasg(\asgrnd \reco \asg_\ttt)$,
$\sq{\asg_1}' \sim \stsqasg(\asg'_\sss \reco \asgrnd)$, and
$\sq{\asg_2}' \sim \stsqasg(\asgrnd \reco \asg'_\ttt)$,
it holds that
\begin{align}
\label{eq:alg:monotone}
    \Pr_{\substack{
        \asgrnd, \sq{\asg_1}, \sq{\asg_2}
    }}\bigl[
        \sq{\asg_1} \circ \sq{\asg_2} \text{ satisfies } C
    \bigr]
    \leq 
    \Pr_{\substack{
        \asgrnd, \sq{\asg_1}', \sq{\asg_2}'
    }}\bigl[
        \sq{\asg_1}' \circ \sq{\asg_2}' \text{ satisfies } C
    \bigr].
\end{align}
\end{claim}
\begin{proof}
It is sufficient to show \cref{eq:alg:monotone} when
$(\asg_\sss, \asg'_\sss)$ and $(\asg_\ttt, \asg'_\ttt)$ differ in a single variable.
Without loss of generality, we can assume that
$\asg_\sss \neq \asg'_\sss$ and $\asg_\ttt = \asg'_\ttt$.
By reordering the $k$ literals,
we can assume that
$\asg_\sss(\ell_1) = 0$,
$\asg'_\sss(\ell_1) = 1$, and
$\asg_\sss(\ell_i) = \asg'_\sss(\ell_i)$ for every $i \neq 1$.
Conditioned on the random assignment $\asgrnd$, we have
\begin{alignat}{3}
    & \Pr_{\sq{\asg_1}, \sq{\asg_2}}\bigr[
        \sq{\asg_1} \circ \sq{\asg_2} \text{ satisfies } C \bigm| \asgrnd
    \bigr]
    && = \Pr_{\sq{\asg_1}}\bigl[
        \sq{\asg_1} \text{ satisfies } C \bigm| \asgrnd
    \bigr]
    && \cdot
    \Pr_{\sq{\asg_2}}\bigl[
        \sq{\asg_2} \text{ satisfies } C \bigm| \asgrnd
    \bigr], \\
    & \Pr_{\sq{\asg_1}', \sq{\asg_2}'}\bigr[
        \sq{\asg_1}' \circ \sq{\asg_2}' \text{ satisfies } C \bigm| \asgrnd
    \bigr]
    && = \Pr_{\sq{\asg_1}'}\bigl[
        \sq{\asg_1}' \text{ satisfies } C \bigm| \asgrnd
    \bigr]
    && \cdot
    \Pr_{\sq{\asg_2}'}\bigl[
        \sq{\asg_2}' \text{ satisfies } C \bigm| \asgrnd
    \bigr].
\end{alignat}
Since $\sq{\asg_2}$ and $\sq{\asg_2}'$ follow the same distribution as $\stsqasg(\asg_\ttt \reco \asgrnd) = \stsqasg(\asg'_\ttt \reco \asgrnd)$, it holds that
\begin{align}
    \Pr_{\sq{\asg_2}}\bigl[\sq{\asg_2} \text{ satisfies } C \bigm| \asgrnd\bigr]
    = \Pr_{\sq{\asg_2}'}\bigl[\sq{\asg_2}' \text{ satisfies } C \bigm| \asgrnd\bigr].
\end{align}
Consider the following case analysis on $\asgrnd(\ell_1)$:

\begin{description}
    \item[(Case 1)] $\asgrnd(\ell_1) = 0$. \\
        Fix any irredundant reconfiguration sequence $\sq{\asg_1}'$ from $\asg'_\sss$ to $\asgrnd$.
        There exists a unique ordering $\sigma'$ over $\asg'_\sss \triangle \asgrnd$ such that
        starting from $\asg'_\sss$,
        we obtain $\sq{\asg_1}'$ by flipping the assignments to $\sigma'(1), \sigma'(2), \ldots,$ in this order.
        Letting $\sigma$ be an ordering over $\asg_\sss \triangle \asgrnd$
        obtained by removing $x_1$ from $\sigma'$,
        we define $\sq{\bsg_1}$
        as an (irredundant) reconfiguration sequence from $\asg_\sss$ to $\asgrnd$
        obtained by flipping the assignments to $\sigma(1), \sigma(2), \ldots,$ in this order.
        By construction,
        if $\sq{\bsg_1}$ satisfies $C$, then $\sq{\asg_1}'$ also satisfies $C$.
        Moreover, if $\sq{\asg_1}'$ is uniformly distributed over $\stsqasg(\asg'_\sss \reco \asgrnd)$,
        then $\sq{\bsg_1}$ is uniformly distributed over $\stsqasg(\asg_\sss \reco \asgrnd)$.
        Therefore, we derive
        \begin{align}
        \begin{aligned}
            \Pr_{\sq{\asg_1}}\bigl[
                \sq{\asg_1} \text{ satisfies } C \bigm| \asgrnd
            \bigr]
            = \Pr_{\sq{\bsg_1}}\bigl[
                \sq{\bsg_1} \text{ satisfies } C \bigm| \asgrnd
            \bigr]
            \leq \Pr_{\sq{\asg_1}'}\bigl[
                \sq{\asg_1}' \text{ satisfies } C \bigm| \asgrnd
            \bigr].
        \end{aligned}
        \end{align}
    \item[(Case 2)] $\asgrnd(\ell_1) = 1$. \\
        Since any irredundant reconfiguration sequence $\sq{\asg_1}'$ from $\asg'_\sss$ to $\asgrnd$
        does not flip $x_1$'s assignment, we have
        \begin{align}
            \Pr_{\sq{\asg_1}'}\bigl[
                \sq{\asg_1}' \text{ satisfies } C \bigm| \asgrnd
            \bigr] 
            = 1 \geq
            \Pr_{\sq{\asg_1}}\bigl[
                \sq{\asg_1} \text{ satisfies } C  \bigm| \asgrnd
            \bigr].
        \end{align}
\end{description}
In either case, it holds that
\begin{align}
    \Pr_{\sq{\asg_1}}\bigl[
        \sq{\asg_1} \text{ satisfies } C  \bigm| \asgrnd
    \bigr]
    \leq
    \Pr_{\sq{\asg_1}'}\bigl[
        \sq{\asg_1}' \text{ satisfies } C \bigm| \asgrnd
    \bigr].
\end{align}
Consequently, we obtain
\begin{align}
    \Pr_{\sq{\asg_1}, \sq{\asg_2}}\bigl[
        \sq{\asg_1} \circ \sq{\asg_2} \text{ satisfies } C \bigm| \asgrnd
    \bigr]
    \leq \Pr_{\sq{\asg_1}', \sq{\asg_2}'}\bigl[
        \sq{\asg_1}' \circ \sq{\asg_2}' \text{ satisfies } C \bigm| \asgrnd
    \bigr],
\end{align}
which implies \cref{eq:alg:monotone}, as desired.
\end{proof}

By \cref{clm:alg:monotone},
it is sufficient to prove \cref{eq:alg:round} only when
both $\asg_\sss$ and $\asg_\ttt$ make a single literal of $C$ true.
Thus, we shall bound the left-hand side of \cref{eq:alg:round} in each case of 
$\asg_\sss \neq \asg_\ttt$ and $\asg_\sss = \asg_\ttt$.
Without loss of generality,
we can safely assume that the clause $C$ is \emph{positive}; i.e., $C = x_1 \vee \cdots \vee x_k$.
Hereafter,
let $\asgrnd \colon V \to \zo$ be a uniformly random assignment, and
let $\sq{\asg_1}$ and $\sq{\asg_2}$ be two random irredundant reconfiguration sequences
uniformly chosen from
$\stsqasg(\asg_\sss \reco \asgrnd)$ and $\stsqasg(\asgrnd \reco \asg_\ttt)$,
respectively.
We will show the following two claims.

\begin{claim}
\label{clm:alg:neq}
Suppose that
$\asg_\sss$ and $\asg_\ttt$ make a single literal of $C$ true and
$\asg_\sss \neq \asg_\ttt$.
Then, it holds that
\begin{align}
\label{eq:alg:neq}
\begin{aligned}
    \Pr_{\asgrnd, \sq{\asg_1}, \sq{\asg_2}}\bigl[
        \sq{\asg_1} \circ \sq{\asg_2} \text{ satisfies } C
    \bigr]
    & = \sum_{0 \leq j \leq K} \frac{1}{4} \cdot \frac{\binom{K}{j}}{2^K} \cdot \left[
        \left(\frac{j}{j+1}\right)^2 + \frac{j+1}{j+2} + \frac{j+1}{j+2} + 1
    \right] \\
    & \geq 1 - \frac{1}{k-1} - \frac{1}{k},
\end{aligned}
\end{align}
where $K \defeq k-2$.
\end{claim}

\begin{claim}
\label{clm:alg:eq}
Suppose that
$\asg_\sss$ and $\asg_\ttt$ make a single literal of $C$ true and
$\asg_\sss = \asg_\ttt$.
Then, it holds that
\begin{align}
\label{eq:alg:eq}
\begin{aligned}
    \Pr_{\asgrnd, \sq{\asg_1}, \sq{\asg_2}}\bigl[
        \sq{\asg_1} \circ \sq{\asg_2} \text{ satisfies } C
    \bigr]
    & = \sum_{0 \leq j \leq K} \frac{1}{2} \cdot \frac{\binom{K}{j}}{2^K} \cdot
    \left[\left(\frac{j}{j+1}\right)^2 + 1\right] \\
    & \geq 1 - \frac{2}{k},
\end{aligned}
\end{align}
where $K \defeq k-1$.
\end{claim}

\cref{lem:alg:round} follows from \cref{clm:alg:monotone,clm:alg:neq,clm:alg:eq}.
\begin{remark}
By numerically evaluating \cref{eq:alg:neq,eq:alg:eq},
we obtain approximation factors better than $1-\frac{1}{k-1}-\frac{1}{k}$ for small $k$,
as shown by \cref{tab:approx} in \cref{subsec:intro:results}.
\end{remark}
In the proof of \cref{clm:alg:neq,clm:alg:eq},
we use the following equality for the sum of binomial coefficients,
whose proof is deferred to \cref{app}.

\begin{fact}[$*$]
\label{fct:alg:binom}
    For any integers $k$ and $n$ with $0 \leq k \leq n$, it holds that
    \begin{align}
        \sum_{0 \leq k \leq n} \binom{n}{k} \frac{1}{k+1} & = \frac{2^{n+1}-1}{n+1}, \\
        \sum_{0 \leq k \leq n} \binom{n}{k} \frac{1}{k+2} & = \frac{2^{n+1} \cdot n + 1}{(n+1)(n+2)}.
    \end{align}
\end{fact}

\begin{proof}[Proof of \cref{clm:alg:neq}]
By reordering the $k$ variables, we can assume that
\begin{align}
    \asg_\sss(x_i) & =
    \begin{cases}
        1 & \text{if } i = k, \\
        0 & \text{otherwise},
    \end{cases} \\
    \asg_\ttt(x_i) & =
    \begin{cases}
        1 & \text{if } i = k-1, \\
        0 & \text{otherwise}.
    \end{cases}
\end{align}
Define $K \defeq k-2$ and
$V_{\leq K} \defeq \{x_1, \ldots, x_K\}$.
Note that $\asg_\sss|_{V_{\leq K}} = \asg_\ttt|_{V_{\leq K}} = 0^K$.
Consider the following case analysis
on $\asgrnd(x_{k-1})$ and $\asgrnd(x_k)$:
\begin{description}
    \item[(Case 1)] $\asgrnd(x_{k-1}) = 0$ and $\asgrnd(x_k) = 0$. \\
        Condition on the number of $1$'s in $\asgrnd|_{V_{\leq K}}$, denoted by $j$,
        which occurs with probability
        $\frac{\binom{K}{j}}{2^K}$.
        Observe that $\sq{\asg_1}$
            satisfies $C$ if and only if it does not flip $x_k$'s assignment at first,
            which happens with probability $1 - \frac{1}{j+1} = \frac{j}{j+1}$.
        Similarly,
            $\sq{\asg_2}$
            satisfies $C$ with probability $\frac{j}{j+1}$.
        Therefore, $\sq{\asg_1} \circ \sq{\asg_2}$
            satisfies $C$
            with probability $\left(\frac{j}{j+1}\right)^2$.
    
    \item[(Case 2)] $\asgrnd(x_{k-1}) = 1$ and $\asgrnd(x_k) = 0$. \\
        Condition on the number of $1$'s in $\asgrnd|_{V_{\leq K}}$, denoted by $j$.
        Then, $\sq{\asg_1}$
            satisfies $C$ if and only if it does not flip $x_k$'s assignment at first,
            which occurs with probability $1 - \frac{1}{j+2} = \frac{j+1}{j+2}$.
        Since $\asg_\ttt(x_{k-1}) = \asgrnd(x_{k-1}) = 1$,
            $\sq{\asg_2}$ satisfies $C$ with probability $1$.
        Therefore, $\sq{\asg_1} \circ \sq{\asg_2}$ satisfies $C$
            with probability $\frac{j+1}{j+2}$.
    
    \item[(Case 3)] $\asgrnd(x_{k-1}) = 0$ and $\asgrnd(x_k) = 1$. \\
        Similarly to (Case~2),
        $\sq{\asg_1} \circ \sq{\asg_2}$ satisfies $C$
        with probability $\frac{j+1}{j+2}$, where
        $j$ is the number of $1$'s in $\asgrnd|_{V_{\leq K}}$.
    
    \item[(Case 4)] $\asgrnd(x_{k-1}) = 1$ and $\asgrnd(x_k) = 1$. \\
        Since $\asg_\sss(x_k) = \asgrnd(x_k) = 1$ and $\asgrnd(x_{k-1}) = \asg_\ttt(x_{k-1}) = 1$,
        we find any irredundant reconfiguration sequence of
            $\stsqasg(\asg_\sss \reco \asgrnd)$ and
            $\stsqasg(\asgrnd \reco \asg_\ttt)$
        to satisfy $C$.
        Therefore, $\sq{\asg_1} \circ \sq{\asg_2}$ satisfies $C$ with probability $1$.
\end{description}
Since each of the above four cases occurs with probability $\frac{1}{4}$,
we derive
\begin{align}
\begin{aligned}
    \Pr_{\asgrnd, \sq{\asg_1}, \sq{\asg_2}}\bigl[
        \sq{\asg_1} \circ \sq{\asg_2} \text{ satisfies } C
    \bigr]
    & = \sum_{0 \leq j \leq K} \frac{1}{4} \cdot \frac{\binom{K}{j}}{2^K} \cdot \left[
        \left(\frac{j}{j+1}\right)^2 + \frac{j+1}{j+2} + \frac{j+1}{j+2} + 1
    \right] \\
    & = 2^{-K} \cdot \sum_{0 \leq j \leq K} \binom{K}{j} \cdot \frac{1}{4} \cdot \left[
        4 - \frac{2}{j+1} - \frac{2}{j+2} + \frac{1}{(j+1)^2}
    \right] \\
    & \geq 2^{-K} \cdot \sum_{0 \leq j \leq K} \binom{K}{j} \cdot \left[
        1 - \frac{1}{2}\cdot\frac{1}{j+1} - \frac{1}{2}\cdot\frac{1}{j+2}
    \right] \\
    & \underbrace{=}_{\text{\cref{fct:alg:binom}}} 2^{-K} \cdot \left[
        2^K - \frac{1}{2}\cdot\frac{2^{K+1} - 1}{K+1} - \frac{1}{2}\cdot\frac{2^{K+1} \cdot K + 1}{(K+1)\cdot(K+2)}
    \right] \\
    & \geq 2^{-K} \cdot \left[
        2^K - \frac{2^K}{K+1} - \frac{2^K}{K+2}
    \right] \\
    & \underbrace{=}_{K = k-2} 1 - \frac{1}{k-1} - \frac{1}{k},
\end{aligned}
\end{align}
which completes the proof.
\end{proof}

\begin{proof}[Proof of \cref{clm:alg:eq}]
By reordering the $k$ variables, we can assume that
\begin{align}
    \asg_\sss(x_i) = \asg_\ttt(x_i) =
    \begin{cases}
        1 & \text{if } i = k, \\
        0 & \text{otherwise}.
    \end{cases}
\end{align}
Define $K \defeq k-1$ and $V_{\leq K} \defeq \{x_1, \ldots, x_K\}$.
Note that $\asg_\sss|_{V_{\leq K}} = \asg_\ttt|_{V_{\leq K}} = 0^K$.
Consider the following case analysis on $\asgrnd(x_k)$:
\begin{description}
    \item[(Case 1)] $\asgrnd(x_k) = 0$. \\
        Condition on the number of $1$'s in $\asgrnd|_{V_{\leq K}}$,
        denoted by $j$, which occurs with probability $\frac{\binom{K}{j}}{2^K}$.
        Observe that $\sq{\asg_1}$
        satisfies $C$ if and only if it does not flip $x_k$'s assignment at first,
        which happens with probability $1 - \frac{1}{j+1} = \frac{j}{j+1}$.
        Similarly, $\sq{\asg_2}$ satisfies $C$ with probability $\frac{j}{j+1}$.
        Therefore, $\sq{\asg_1} \circ \sq{\asg_2}$ satisfies $C$ with probability
        $\left(\frac{j}{j+1}\right)^2$.
    
    \item[(Case 2)] $\asgrnd(x_k) = 1$. \\
        Since $\asg_\sss(x_k) = \asgrnd(x_k) = 1$ and
        $\asgrnd(x_k) = \asg_\ttt(x_k) = 1$,
        we find any irredundant reconfiguration of $\stsqasg(\asg_\sss \reco \asgrnd)$ and
        $\stsqasg(\asgrnd \reco \asg_\ttt)$ to satisfy $C$.
        Therefore, $\sq{\asg_1} \circ \sq{\asg_2}$ satisfies $C$ with probability $1$.
\end{description}
Since each of the above two cases occurs with probability $\frac{1}{2}$, we derive
\begin{align}
\begin{aligned}
    \Pr_{\asgrnd, \sq{\asg_1}, \sq{\asg_2}}\bigl[
        \sq{\asg_1} \circ \sq{\asg_2} \text{ satisfies } C
    \bigr]
    & = \sum_{0 \leq j \leq K} \frac{1}{2} \cdot \frac{\binom{K}{j}}{2^K} \cdot
    \left[\left(\frac{j}{j+1}\right)^2 + 1\right] \\
    & = 2^{-K} \cdot \sum_{0 \leq j \leq K} \binom{K}{j} \cdot \frac{1}{2} \cdot
    \left[ 2 - 2 \cdot \frac{1}{j+1} + \frac{1}{(j+1)^2} \right] \\
    & \geq 2^{-K} \cdot \sum_{0 \leq j \leq K} \binom{K}{j} \cdot
    \left[1 - \frac{1}{j+1}\right] \\
    & \underbrace{=}_{\text{\cref{fct:alg:binom}}} 2^{-K} \cdot \left[2^K - \frac{2^{K+1}-1}{K+1}\right] \\
    & \geq 2^{-K} \cdot \left[2^K - \frac{2^{K+1}}{K+1}\right] \\
    & \underbrace{=}_{K = k-1} 1 - \frac{2}{k},
\end{aligned}    
\end{align}
which completes the proof.
\end{proof}

\section{$\PSPACE$-hardness of \texorpdfstring{$\left(1-\frac{3-\epsilon}{28k}\right)$-factor}{(1-(3-ε)/28k)-factor} Approximation
of \MMkSATReconf for Large $k$}
\label{sec:hard9.333}

In this section, we prove \cref{thm:intro:hard9.333}; i.e.,
\MMkSATReconf is $\PSPACE$-hard to approximate within a factor of $1-\frac{3-\epsilon}{28k}$ for every sufficiently large $k$.

\begin{theorem}
\label{thm:hard9.333}
For any real $\epsilon > 0$,
there exists an integer $k_0(\epsilon) \in \bbN$ such that
for any integer $k \geq k_0(\epsilon)$,
\prb{Gap$_{1,1-\frac{3-\epsilon}{28k}}$ \kSATReconf} is $\PSPACE$-hard.
In particular,
\MMkSATReconf is $\PSPACE$-hard to approximate within a factor of 
$1-\frac{3-\epsilon}{28k}$ for every integer $k \geq k_0(\epsilon)$.
\end{theorem}

An an immediate corollary of \cref{thm:hard9.333}, we obtain the $\PSPACE$-hardness of
$\left(1-\Omega\left(\frac{1}{k}\right)\right)$-factor approximation for every $k \geq 3$,
which is proved in \cref{app} for the sake of completeness.

\begin{corollary}[$*$]
\label{cor:hard0}
    There exists a universal constant $\delta_0 > 0$ such that
    for any integer $k \geq 3$,
    \prb{Gap$_{1,1-\frac{\delta_0}{k}}$ \kSATReconf} is $\PSPACE$-hard.
    In particular,
    \MMkSATReconf is $\PSPACE$-hard to approximate
    within a factor of $1-\frac{\delta_0}{k}$ for every integer $k \geq 3$.
\end{corollary}

\subsection{Outline of the Proof of \texorpdfstring{\cref{thm:hard9.333}}{Theorem~\ref{thm:hard9.333}}}

We present an outline of the proof of \cref{thm:hard9.333}.
Starting from a PCRP system for $\PSPACE$ whose query complexity is $q$,
we reduce it to \prb{Maxmin E$k$-SAT Reconfiguration}
for any sufficiently large integer $k \geq q \cdot \rep_0$,
where $\rep_0$ depends only on the parameters of the PCRP system,
with the following properties.

\begin{lemma}
\label{lem:hard9.333:main}
Suppose that there exists a PCRP system $(\V_L, \pf_\sss, \pf_\ttt)$
for a $\PSPACE$-complete language $L \subseteq \zo^*$, where 
$\V_L$ is a verifier with
randomness complexity $r(n) = \Theta(\log n)$,
query complexity $q(n) = q \geq 3$,
perfect completeness $c(n) = 1$,
soundness $s(n) = s \in (0,1)$, and
degree $\Delta(n) = \Delta \in \bbN$, and
$\pf_\sss,\pf_\ttt \colon \allowbreak \zo^* \to \zo^*$ are polynomial-time computable proofs.
Then, for any real $\epsilon \in (0,1)$,
there exists an integer $\rep_0(\epsilon,s,q) \in \bbN$ such that
for any integer $k \geq q \cdot \rep_0(\epsilon,s,q)$,
there exists a polynomial-time reduction that
takes an input $x \in \zo^*$ for $L$ and
returns an instance $(\phi, \asg_\sss, \asg_\ttt)$ of
\prb{Maxmin E$k$-SAT Reconfiguration} such that the following hold\textup{:}
\begin{itemize}
    \item \textup{(}Completeness\textup{)}
        If $x \in L$, then $\opt_\phi\bigl(\asg_\sss \reco \asg_\ttt\bigr) = 1$.
    \item \textup{(}Soundness\textup{)}
        If $x \notin L$, then $\opt_\phi\bigl(\asg_\sss \reco \asg_\ttt\bigr) < 1-\zeta$, where
    \begin{align}
        \zeta \defeq
        \frac{1}{(2^q-1) \cdot k} \cdot \left(\frac{q}{4}-\epsilon\right).
    \end{align}
\end{itemize}
In particular,  \prb{Gap$_{1,1-\zeta}$ E$k$-SAT Reconfiguration} is $\PSPACE$-hard.
\end{lemma}
\noindent
By using \cref{lem:hard9.333:main}, we can prove \cref{thm:hard9.333}.

\begin{proof}[Proof of \cref{thm:hard9.333}]
By the PCRP theorem \cite{hirahara2024probabilistically,karthikc.s.2023inapproximability}
and gap-preserving reductions of \cite[Theorem~3.1]{ohsaka2023gap},
\prb{Gap$_{1,s}$ \treSATReconf} is $\PSPACE$-complete
for some real $s \in (0,1)$
even when each variable appears in at most $\Delta$ clauses for some integer $\Delta \in \bbN$.
Let $q \defeq 3$, and
$(\V, \pf_\sss, \pf_\ttt)$ be a PCRP system corresponding to \prb{Gap$_{1,s}$ \treSATReconf},
where $\V$ has
randomness complexity $r(n) = \Theta(\log n)$,
query complexity $q(n) = q$,
perfect completeness $c(n) = 1$,
soundness $s(n) = s$, and
degree $\Delta(n) = \Delta$.
For any real $\epsilon > 0$, let
\begin{align}
    \bar{\epsilon} & \defeq \frac{\epsilon}{100}, \\
    k_0(\epsilon) & \defeq q \cdot \rep_0(\bar{\epsilon},s,q),
\end{align}
where $\rep_0(\bar{\epsilon},s,q)$ is as defined in \cref{lem:hard9.333:main}.
For any integer $k \geq k_0(\epsilon)$,
we apply \cref{lem:hard9.333:main} to $\V$
and deduce that
\prb{Gap$_{1,1-\zeta}$ E$k$-SAT Reconfiguration} is $\PSPACE$-hard,
where $\zeta$ is calculated as
\begin{align}
    \zeta \defeq
    \frac{1}{(2^q-1) \cdot k} \cdot \left(\frac{q}{4}-\bar{\epsilon}\right)
    \underbrace{=}_{q=3} \frac{1}{7 \cdot k} \cdot \left(\frac{3}{4}-\bar{\epsilon}\right).
    \underbrace{\geq}_{\bar{\epsilon} = \frac{\epsilon}{100}}
    \frac{3-\epsilon}{28\cdot k}
\end{align}
which accomplishes the proof.
\end{proof}

The remainder of this section is devoted to the proof of \cref{lem:hard9.333:main}.

\subsection{Proof of \texorpdfstring{\cref{lem:hard9.333:main}}{Lemma~\ref{lem:hard9.333:main}}}
Let $(\V_L, \pf_\sss, \pf_\ttt)$ be a PCRP system for a $\PSPACE$-complete language $L \subseteq \zo^*$,
where $\V_L$ is a verifier with
randomness complexity $r(n) = \Theta(\log n)$,
query complexity $q(n) = q \geq 3$,
perfect completeness $c(n) = 1$,
soundness $s(n) = s \in (0,1)$, and
degree $\Delta(n) = \Delta \in \bbN$, and
$\pf_\sss,\pf_\ttt \colon \zo^* \to \zo^*$ are
polynomial-time computable proofs.
We can safely assume that 
any possible query sequence generated by $\V_L$ contains \emph{exactly} $q$ locations.
Let $\gap \defeq 1-s \in (0,1)$ and $\bal \defeq \frac{q}{2}$.
For any real $\epsilon \in (0,1)$,
we define $\delta \defeq \frac{\epsilon}{4}$ and 
\begin{align}
\label{eq:hard9.333:rep_0}
    \rep_0(\epsilon,s,q) \defeq
    \left\lceil \frac{\bal \cdot (\bal+\delta)}{\delta} \cdot \frac{1}{\gap \cdot q} \right\rceil,
\end{align}
which depends only on $\epsilon$, $s$, and $q$.\footnote{
The choice of $\rep_0(\epsilon,s,q)$ will be crucial in the proof of \cref{lem:hard9.333:Horn}.
}
For any integer $k \geq q \cdot \rep_0(\epsilon,s,q)$,
we define $\rep \defeq \left\lfloor \frac{k}{q} \right\rfloor$.
By definition, $q\rep \leq k \leq q\rep + q-1$.
Let $x \in \zo^n$ be an input for $L$.
The proof length for $\V_L(x)$ is denoted by $\ell(n)$, which is polynomially bounded in $n$.
Let $\pf_\sss \defeq \pf_\sss(x)$ and $\pf_\ttt \defeq \pf_\ttt(x)$ be
the starting and ending proofs in $\zo^{\ell(n)}$ associated with $\V_L(x)$, respectively.
Note that $\V_L(x)$ accepts both $\pf_\sss$ and $\pf_\ttt$ with probability $1$.
Moreover, the following hold:
\begin{itemize}
    \item (Completeness) If $x \in L$, then $\opt_{\V_L(x)}\bigl(\pf_\sss \reco \pf_\ttt\bigr) = 1$.
    \item (Soundness) If $x \notin L$, then $\opt_{\V_L(x)}\bigl(\pf_\sss \reco \pf_\ttt\bigr) < 1-\gap$.
\end{itemize}
Hereafter, we will assume without loss of generality that
the input length $n$ is sufficiently large so that\footnote{
Such an integer $n$ always exists since 
the left-hand sides of \cref{eq:hard9.333:n} decrease as $n$ increases, while
the right-hand sides are constants.
}
\begin{align}
\label{eq:hard9.333:n}
\begin{aligned}
    \frac{\Delta}{2^{r(n)}} + \frac{q}{\ell(n)}
    & \leq \frac{\delta}{k}, \\
    q\cdot(\rep+1)^2 \cdot \left(\frac{\Delta}{2^{r(n)}} + \frac{q}{\ell(n)}\right)
    & \leq \frac{\delta}{k} \cdot \left(1-\frac{\bal+\delta}{q}\right).
\end{aligned}
\end{align}
In the subsequent sections, we introduce several verifiers using $\V_L$ and analyze their completeness and soundness.

\subsubsection{All-One Verifier}
The first verifier is the \emph{all-one verifier} $\A_p$.
Given an integer $p \in \bbN$ and oracle access to a proof $\qf \in \zo^{\ell(n)}$,
$\A_p$ samples a query sequence $I$ of $p$ distinct locations from $[\ell(n)]$ and
accepts if $\qf(i) = 1$ for every location $i \in I$
(i.e., $\qf|_I = 1^p$),
as described below.
\begin{itembox}[l]{\textbf{$p$-query all-one verifier $\A_p$}}
\begin{algorithmic}[1]
    \item[\textbf{Input:}]
        an integer $p \in \bbN$.
    \item[\textbf{Oracle access:}]
        a proof $\qf \in \zo^{\ell(n)}$.
    \State sample a query sequence $I$ from $\binom{[\ell(n)]}{p}$.
    \If{$\qf(i) = 1$ for every $i \in I$}
        \State \Return $1$.
    \Else
        \State \Return $0$.
    \EndIf
\end{algorithmic}
\end{itembox}

\noindent
Observe that
$\A_p$ has the randomness complexity at most $p \cdot \log \ell(n) = \Theta(\log n)$, and
$\A_p$ always generates a fixed circuit $D \colon \zo^p \to \zo$
that accepts only $1^p$ (i.e., $D(f) \defeq \llbracket f = 1^p \rrbracket$).
Note also that $\A_p$'s rejection probability is monotonically increasing in $p$; namely,
\begin{align}
    \Pr\bigl[\A_{p+1} \text{ rejects } \qf\bigr] \geq
    \Pr\bigl[\A_p \text{ rejects } \qf\bigr].
\end{align}

\subsubsection{Combined Verifier}
The second verifier is the $q$-query \emph{combined verifier} $\W$.
Given oracle access to a pair of proofs, denoted by $\Pf \defeq \pf \circ \qf \in \zo^{2\ell(n)}$,
$\W$
calls $\V_L(x)$ on $\pf$ with probability $\frac{\bal}{\gap \cdot k}$ and
calls $\A_q$ on $\qf$ with probability $1-\frac{\bal}{\gap \cdot k}$,
as described below.

\begin{itembox}[l]{\textbf{$q$-query combined verifier $\W$}}
\begin{algorithmic}[1]
    \item[\textbf{Input:}]
        the PCRP verifier $\V_L$,
        the all-one verifier $\A_p$, and
        an input $x \in \zo^n$.
    \item[\textbf{Oracle access:}]
        a proof $\Pf = \pf \circ \qf \in \zo^{2\ell(n)}$.
    \State uniformly sample a real $r \sim (0,1)$.
    \If{$r < \frac{\bal}{\gap \cdot k}$} \Comment{with probability $\frac{\bal}{\gap \cdot k}$}
        \State run $\V_L(x)$ on $\pf$.
        \State \Return $\V_L(x)$'s return value.
    \Else \Comment{with probability $1-\frac{\bal}{\gap \cdot k}$}
        \State run $\A_q$ on $\qf$.
        \State \Return $\A_q$'s return value.
    \EndIf
\end{algorithmic}
\end{itembox}

\noindent
Since $\frac{\bal}{\gap \cdot k} \in (0,1)$ due to \cref{eq:hard9.333:rep_0}, the probabilistic behavior of $\W$ is well defined.
Observe that the randomness complexity of $\W$ is bounded by those of $\V_L$ and $\A_q$; i.e., $\Theta(\log n)$.
The starting and ending proofs $\Pf_\sss, \Pf_\ttt \in \zo^{2\ell(n)}$ are defined as
$\Pf_\sss \defeq \pf_\sss \circ 1^{\ell(n)}$ and
$\Pf_\ttt \defeq \pf_\ttt \circ 1^{\ell(n)}$,
respectively.
Since
$\V_L$ accepts $\pf_\sss$ and $\pf_\ttt$ with probability $1$ and
$\A_q$ accepts $1^{\ell(n)}$ with probability $1$,
$\W$ accepts $\Pf_\sss$ and $\Pf_\ttt$ with probability $1$.
We show the following completeness and soundness.

\begin{lemma}
\label{lem:hard9.333:W}
The following hold\textup{:}
\begin{itemize}
    \item \textup{(}Completeness\textup{)}
        If $\opt_{\V_L(x)}\bigl(\pf_\sss \reco \pf_\ttt\bigr) = 1$,
        then $\opt_{\W}\bigl(\Pf_\sss \reco \Pf_\ttt\bigr) = 1$.
    \item \textup{(}Soundness\textup{)}
        If $\opt_{\V_L(x)}\bigl(\pf_\sss \reco \pf_\ttt\bigr) < 1-\gap$,
        then $\opt_{\W}\bigl(\Pf_\sss \reco \Pf_\ttt\bigr)  < 1-\frac{\bal}{k}$.
    Moreover, for any reconfiguration sequence
    $\sq{\Pf} = (\Pf^{(1)}, \ldots, \Pf^{(T)})$ from $\Pf_\sss$ to $\Pf_\ttt$,
    there exists a proof $\Pf^{(t)}$ in $\sq{\Pf}$ such that
    \begin{align}
        1-\frac{\bal}{k}
        \leq \Pr\bigl[\W \text{ accepts } \Pf^{(t)}\bigr]
        \leq 1 - \frac{\bal-\delta}{k}.
    \end{align}
\end{itemize}
\end{lemma}

To prove \cref{lem:hard9.333:W}, we use the following claim.
\begin{claim}
\label{clm:hard9.333:W}
    Each proof location $\qf(i)$ is queried by $\A_p$ with probability $\frac{p}{\ell(n)}$.
    Each proof location $\Pf(i)$ is queried by $\W$ with probability at most
    \begin{align}
        \frac{\Delta}{2^{r(n)}} + \frac{q}{\ell(n)}.
    \end{align}
\end{claim}
\begin{proof} 
The former statement holds by the definition of $\A_p$.
Since $\pf(i)$ is queried by $\W$ only if $\V_L$ is called, we have
\begin{align}
    \Pr\bigl[ \W \text{ queries } \pf(i) \bigr]
    \leq \frac{1}{\gap \cdot k} \cdot \frac{\Delta}{2^{r(n)}}
    \leq \frac{\Delta}{2^{r(n)}}.
\end{align}
Since $\qf(i)$ is queried by $\W$ only if $\A_q$ is called,
we have
\begin{align}
    \Pr\bigl[ \W \text{ queries } \qf(i) \bigr]
    \leq \left(1-\frac{1}{\gap \cdot k}\right) \cdot \frac{q}{\ell(n)}
    \leq \frac{q}{\ell(n)}.
\end{align}
Consequently, any location of $\Pf$ is queried by $\W$ with probability at most
\begin{align}
    \max\left\{ \frac{\Delta}{2^{r(n)}}, \frac{q}{\ell(n)} \right\}
    \leq \frac{\Delta}{2^{r(n)}} + \frac{q}{\ell(n)},
\end{align}
as desired.
\end{proof}

\begin{proof}[Proof of \cref{lem:hard9.333:W}]
We first show the completeness.
Suppose that $\opt_{\V_L(x)}(\pf_\sss \reco \pf_\ttt) = 1$.
Let $\sq{\pf} = (\pf^{(1)}, \ldots, \pf^{(T)})$ be a reconfiguration
sequence from $\pf_\sss$ to $\pf_\ttt$ such that
$\val_{\V_L(x)}(\sq{\pf}) = 1$.
Constructing a reconfiguration sequence
$\sq{\Pf} = (\Pf^{(1)}, \ldots, \Pf^{(T)})$ from $\Pf_\sss$ to $\Pf_\ttt$ 
such that $\Pf^{(t)} \defeq \pf^{(t)} \circ 1^{\ell(n)}$ for every $t \in [T]$,
we find $\W$ to accept every proof $\Pf^{(t)}$ with probability $1$,
implying that
$\opt_{\W}(\Pf_\sss \reco \Pf_\ttt) = 1$, as desired.

We next show the soundness.
Suppose that $\opt_{\V_L(x)}(\pf_\sss \reco \pf_\ttt) < 1-\gap$.
Let $\sq{\Pf} = (\pf^{(1)} \circ \qf^{(1)}, \ldots, \allowbreak \pf^{(T)} \circ \qf^{(T)})$
be any reconfiguration sequence from $\Pf_\sss$ to $\Pf_\ttt$ such that
$\val_{\W}(\sq{\Pf}) = \opt_{\W}(\Pf_\sss \reco \Pf_\ttt)$.
Since
$\sq{\pf} = (\pf^{(1)}, \ldots, \pf^{(T)})$
is a reconfiguration sequence from $\pf_\sss$ to $\pf_\ttt$,
we have $\val_{\V_L(x)}(\sq{\pf}) < 1-\gap$ by assumption; in particular,
there exists a proof $\pf^{(t)}$ in $\sq{\pf}$ such that
$\Pr[\V_L(x) \text{ rejects } \pf^{(t)}] > \gap$.
Since $\W$ calls $\V_L(x)$ with probability $\frac{\bal}{\gap \cdot k}$, we have
$\Pr[\W \text{ rejects } \Pf^{(t)}] > \frac{\bal}{k}$, implying that
$\opt_{\W}(\Pf_\sss \reco \Pf_\ttt) < 1-\frac{\bal}{k}$, as desired.

We finally show the ``moreover'' part.
Let $\sq{\Pf}$ be any reconfiguration sequence from $\Pf_\sss$ to $\Pf_\ttt$.
By the soundness shown above,
$\sq{\Pf}$ contains an adjacent pair of proofs,
denoted by $\Pf^\circ$ and $\Pf'$, such that
$\val_{\W}(\Pf^\circ) \geq 1-\frac{\bal}{k}$ and
$\val_{\W}(\Pf') < 1-\frac{\bal}{k}$.
Since $\Pf^\circ$ and $\Pf'$ differ in a single location,
which is queried by $\W$ with probability at most $\frac{\Delta}{2^{r(n)}} + \frac{q}{\ell(n)}$
due to \cref{clm:hard9.333:W},
we have
\begin{align}
    \bigl| \val_{\W}(\Pf^\circ) - \val_{\W}(\Pf') \bigr|
    \leq \frac{\Delta}{2^{r(n)}} + \frac{q}{\ell(n)},
\end{align}
Consequently, we derive
\begin{align}
    & \val_{\W}(\Pf^\circ)
    \leq \val_{\W}(\Pf') + \frac{\Delta}{2^{r(n)}} + \frac{q}{\ell(n)}
    \leq 1-\frac{\bal}{k} +  \frac{\Delta}{2^{r(n)}} + \frac{q}{\ell(n)}
    \underbrace{\leq}_{\text{\cref{eq:hard9.333:n}}} 1 - \frac{\bal-\delta}{k} \\
    & \implies 
    1-\frac{\bal}{k}
    \leq \Pr\bigl[\W \text{ accepts } \Pf^\circ\bigr]
    \leq 1-\frac{\bal-\delta}{k},
\end{align}
which completes the proof.
\end{proof}

\subsubsection{Horn Verifier}

Consider now the $k$-query \emph{Horn verifier} $\Vhorn$ described below.
Given oracle access to a proof $\Pf \in \zo^{2\ell(n)}$,
$\Vhorn$ generates
$(I_1, D_1)$ from $\W$,
$(I_2,D_2), \ldots, (I_{\rep},D_{\rep})$ from $\A_q$, and
$(I_{\rep+1},D_{\rep+1})$ from $\A_{k-q\rep}$, and
accepts if $I_1, \ldots, I_{\rep+1}$ are \emph{not} pairwise disjoint or
the following Horn-like condition holds:
\begin{align}
\label{eq:hard9.333:Horn:accept}
    \bigl( D_1(\Pf|_{I_1}) = 1 \bigr) \vee
    \bigl( D_2(\Pf|_{I_2}) = 0 \bigr)
    \vee \cdots \vee
    \bigl( D_{\rep+1}(\Pf|_{I_{\rep+1}}) = 0 \bigr).
\end{align}
\begin{itembox}[l]{\textbf{$k$-query Horn verifier $\Vhorn$}}
\begin{algorithmic}[1]
    \item[\textbf{Input:}]
        the all-one verifier $\A_p$,
        the combined verifier $\W$, and
        an input $x \in \zo^n$.
    \item[\textbf{Oracle access:}]
        a proof $\Pf = \pf\circ\qf \in \zo^{2\ell(n)}$.
    \State sample a random bit string $R_1 \sim \zo^{\Theta(\log n)}$ used by $\W$ uniformly at random.
    \State run $\W$ on $R_1$ to generate
    a query sequence $I_1$ and
    a circuit $D_1 \colon \zo^q \to \zo$.
    \For{\textbf{each} $2 \leq i \leq \rep+1$}
        \If{$2 \leq i \leq \rep$}
            \State sample a random bit string $R_i \sim \zo^{\Theta(\log n)}$ used by $\A_q$ uniformly at random.
            \State run $\A_q$ on $R_i$ to generate
                a query sequence $I'_i$ and
                a circuit $D_i \colon \zo^q \to \zo$.
        \Else
            \State sample a random bit string $R_i \sim \zo^{\Theta(\log n)}$ used by $\A_{k-q\rep}$ uniformly at random.
            \State run $\A_{k-q\rep}$ on $R_i$ to generate
            a query sequence $I'_i$ and
            a circuit $D_i \colon \zo^{k-q\rep} \to \zo$.
        \EndIf
        \State let $I_i$ be a query sequence
            obtained by shifting $I'_i$ by $\ell(n)$ locations so that
            $\Pf|_{I_i} = \qf|_{I'_i}$.
            \LComment{this step is required since
            $I'_i \subseteq [\ell(n)]$ while
            $\Pf|_{[\ell(n)]} = \pf$.
            }
    \EndFor
    \If{$I_1, \ldots, I_{\rep+1}$ are \emph{not} pairwise disjoint}
        \State \Return $1$.
    \ElsIf{$(D_1(\Pf|_{I_1}) = 1) \vee (D_2(\Pf|_{I_2}) = 0) \vee \cdots \vee (D_{\rep+1}(\Pf|_{I_{\rep+1}}) = 0)$}
        \State \Return $1$.
    \Else
        \State \Return $0$.
    \EndIf
\end{algorithmic}
\end{itembox}

\noindent
The randomness complexity of $\Vhorn$ is at most
$(\rep+1) \cdot \Theta(\log n) = \Theta(\log n)$, and
$\Vhorn$ queries \emph{exactly} $k$ locations of $\Pf$ whenever
$I_1, \ldots, I_{\rep+1}$ are pairwise disjoint because
$|I_1| = \cdots = |I_{\rep}| = q$ and $|I_{\rep+1}| = k-q\rep$.
We show the following completeness and soundness.

\begin{lemma}
\label{lem:hard9.333:Horn}
The following hold\textup{:}
\begin{itemize}
    \item \textup{(}Completeness\textup{)}
    If $\opt_{\V_L(x)}\bigl(\pf_\sss \reco \pf_\ttt\bigr) = 1$, then
    $\opt_{\Vhorn}\bigl(\Pf_\sss \reco \Pf_\ttt\bigr) = 1$.
    \item \textup{(}Soundness\textup{)}
    If $\opt_{\V_L(x)}\bigl(\pf_\sss \reco \pf_\ttt\bigr) < 1-\gap$, then
    $\opt_{\Vhorn}\bigl(\Pf_\sss \reco \Pf_\ttt\bigr) < 1 - \frac{1}{k} \cdot \left(\frac{q}{4}-\epsilon\right)$.
\end{itemize}
\end{lemma}
\begin{proof} 
We first show the completeness.
Suppose that $\opt_{\V_L(x)}(\pf_\sss \reco \pf_\ttt) = 1$.
By \cref{lem:hard9.333:W}, we have
$\opt_{\W}(\Pf_\sss \reco \Pf_\ttt) = 1$.
By the definition of $\Vhorn$,
for any reconfiguration sequence $\sq{\Pf}$ from $\Pf_\sss$ to $\Pf_\ttt$,
it holds that $\val_{\Vhorn}(\sq{\Pf}) \geq \val_{\W}(\sq{\Pf})$, which implies
$\opt_{\Vhorn}(\Pf_\sss \reco \Pf_\ttt) = 1$, as desired.

We next show the soundness.
Suppose that $\opt_{\V_L(x)}(\pf_\sss \reco \pf_\ttt) < 1-\gap$.
Let $\sq{\Pf}$ be any reconfiguration sequence from $\Pf_\sss$ to $\Pf_\ttt$ such that
$\val_{\Vhorn}(\sq{\Pf}) = \opt_{\Vhorn}(\Pf_\sss \reco \Pf_\ttt)$.
By \cref{lem:hard9.333:W}, $\sq{\Pf}$ contains a proof
$\Pf^\circ = \pf^\circ \circ \qf^\circ$ such that
\begin{align}
\label{eq:hard9.333:Horn:Pfcirc}
    1-\frac{\bal}{k}
    \leq \Pr\bigl[ \W \text{ accepts } \Pf^\circ \bigr]
    \leq 1-\frac{\bal-\delta}{k}.
\end{align}
We shall estimate $\Vhorn$'s rejection probability on $\Pf^\circ$.
Since $R_1, \ldots, R_{\rep+1}$ are mutually independent,
we derive the probability that \cref{eq:hard9.333:Horn:accept} does not hold as follows:
\begin{align}
\label{eq:hard9.333:Horn:vee}
\begin{aligned}
    & \Pr_{R_1,\ldots,R_{\rep+1}}\Bigl[
        \bigl( D_1(\Pf^\circ|_{I_1}) = 1 \bigr) \vee
        \bigl( D_2(\Pf^\circ|_{I_2}) = 0 \bigr)
        \vee \cdots \vee
        \bigl( D_{\rep+1}(\Pf^\circ|_{I_{\rep+1}}) = 0 \bigr)
        \text{ is \emph{not} true}
    \Bigr] \\
    & = \Pr_{R_1,\ldots,R_{\rep+1}}\Bigl[
        \bigl( D_1(\Pf^\circ|_{I_1}) = 0 \bigr) \wedge
        \bigl( D_2(\Pf^\circ|_{I_2}) = 1 \bigr)
        \wedge \cdots \wedge
        \bigl( D_{\rep+1}(\Pf^\circ|_{I_{\rep+1}}) = 1 \bigr)
    \Bigr] \\
    & = \Pr_{R_1}\bigl[ D_1(\Pf^\circ|_{I_1}) = 0 \bigr] \cdot
    \left(
        1 - \Pr_{R_2, \ldots, R_{\rep+1}}\Bigl[
            \bigl( D_2(\Pf^\circ|_{I_2}) = 0 \bigr)
            \vee \cdots \vee
            \bigl( D_{\rep+1}(\Pf^\circ|_{I_{\rep+1}}) = 0 \bigr)
        \Bigr]
    \right) \\
    & \geq \Pr_{R_1}\bigl[ D_1(\Pf^\circ|_{I_1}) = 0 \bigr] \cdot
    \left(1 - \sum_{2 \leq i \leq \rep+1}\Pr_{R_i}\bigl[ D_i(\Pf^\circ|_{I_i}) = 0 \bigr] \right) \\
    & \geq \Pr\bigl[ \W \text{ rejects } \Pf^\circ \bigr] \cdot 
        \Bigl(
            1 - \rep \cdot \Pr\bigl[\A_q \text{ rejects } \qf^\circ \bigr]
        \Bigr),
\end{aligned}
\end{align}
where the last inequality used the fact that
\begin{align}
    \Pr_{R_{\rep+1}}\bigl[ D_{\rep+1}(\Pf^\circ|_{I_{\rep+1}}) = 0 \bigr]
    = \Pr\bigl[\A_{k-q\rep} \text{ rejects } \qf^\circ \bigr]
    \underbrace{\leq}_{k-q\rep \leq q} \Pr\bigl[\A_q \text{ rejects } \qf^\circ \bigr].
\end{align}

In the last line of \cref{eq:hard9.333:Horn:vee},
$\W$'s rejection probability is bounded from below by \cref{eq:hard9.333:Horn:Pfcirc}, whereas
$\A_q$'s rejection probability is bounded from above by the following claim.
\begin{claim}
\label{clm:hard9.333:Horn:A}
    It holds that
    \begin{align}
        \Pr\bigl[\A_q \text{ rejects } \qf^\circ\bigr]
        \leq \frac{\bal}{k-\frac{\bal}{\gap}}.
    \end{align}
\end{claim}
\begin{proof} 
Suppose that $\Pr[\A_q \text{ rejects } \qf^\circ] > \frac{\bal}{k-\frac{\bal}{\gap}}$
for contradiction.
Then, we have
\begin{align}
\begin{aligned}
    \Pr\bigl[ \W \text{ accepts } \Pf^\circ \bigr]
    & = \frac{\bal}{\gap \cdot k} \cdot \Pr\bigl[\V_L(x) \text{ accepts } \pf^\circ\bigr]
    + \left(1-\frac{\bal}{\gap \cdot k}\right) \cdot \Pr\bigl[\A_q \text{ accepts } \qf^\circ\bigr] \\
    & < \frac{\bal}{\gap \cdot k} \cdot 1
    + \left(1-\frac{\bal}{\gap \cdot k}\right) \cdot \left(1-\frac{\bal}{k-\frac{\bal}{\gap}}\right) \\
    & = 1 - \left(1-\frac{\bal}{\gap \cdot k}\right) \cdot \frac{\bal}{k-\frac{\bal}{\gap}} \\
    & = 1-\frac{\bal}{k}.
\end{aligned}
\end{align}
On the other hand,
$\Pr[\W \text{ accepts } \Pf^\circ] \geq 1-\frac{\bal}{k}$ by \cref{eq:hard9.333:Horn:Pfcirc},
which is a contradiction.
\end{proof}

By the definition of $\rep_0(\epsilon,s,q)$ in \cref{eq:hard9.333:rep_0}, we have
\begin{align}
    & k \geq q \cdot \rep_0(\epsilon,s,q)
    \underbrace{\geq}_{\text{\cref{eq:hard9.333:rep_0}}} \frac{\bal\cdot(\bal+\delta)}{\delta} \cdot \frac{1}{\gap} \\
    & \implies \frac{\bal}{k-\frac{\bal}{\gap}} \leq \frac{\bal+\delta}{k}.
    \label{eq:hard9.333:Horn:k}
\end{align}
Combining \cref{eq:hard9.333:Horn:Pfcirc,eq:hard9.333:Horn:vee,clm:hard9.333:Horn:A,eq:hard9.333:Horn:k},
we obtain
\begin{align}
\begin{aligned}
    & \Pr_{R_1,\ldots,R_{\rep+1}}\Bigl[
        \bigl( D_1(\Pf^\circ|_{I_1}) = 1 \bigr) \vee
        \bigl( D_2(\Pf^\circ|_{I_2}) = 0 \bigr)
        \vee \cdots \vee
        \bigl( D_{\rep+1}(\Pf^\circ|_{I_{\rep+1}}) = 0 \bigr)
    \Bigr] \\
    & = 1 - \Pr_{R_1,\ldots,R_{\rep+1}}\Bigl[
        \bigl( D_1(\Pf^\circ|_{I_1}) = 1 \bigr) \vee
        \bigl( D_2(\Pf^\circ|_{I_2}) = 0 \bigr)
        \vee \cdots \vee
        \bigl( D_{\rep+1}(\Pf^\circ|_{I_{\rep+1}}) = 0 \bigr)
        \text{ is \emph{not} true}
    \Bigr] \\
    & \leq 1 - \underbrace{\Pr\bigl[ \W \text{ rejects } \Pf^\circ \bigr]}_{\geq \frac{\bal-\delta}{k}} \cdot 
        \Bigl(
            1 - \rep \cdot \underbrace{\Pr\bigl[\A_q \text{ rejects } \qf^\circ \bigr]}_{\leq \frac{\bal}{k-\frac{\bal}{\gap}}}
        \Bigr) \\
    & \leq 1 - \frac{\bal-\delta}{k} \cdot \Biggl(
        1-\rep \cdot \underbrace{\frac{\bal}{k-\frac{\bal}{\gap}}}_{\leq \frac{\bal+\delta}{k}}
    \Biggr) \\
    & \underbrace{\leq}_{k \geq q\rep} 1 - \frac{\bal-\delta}{k} \cdot \left(1-\frac{\bal+\delta}{q}\right).
\end{aligned}
\end{align}

Observe also that
$I_1, \ldots, I_{\rep+1}$ are \emph{not} pairwise disjoint with probability
\begin{align}
\begin{aligned}
    \Pr_{R_1,\ldots,R_{\rep+1}}\bigl[
        I_1, \ldots, I_{\rep+1} \text{ are not pairwise disjoint}
    \bigr]
    & \leq \sum_{i \neq j}
    \Pr_{R_i,R_j}\bigl[
        I_i \text{ and } I_j \text{ are not disjoint}
    \bigr] \\
    & \leq \sum_{i \neq j} q \cdot \left(\frac{\Delta}{2^{r(n)}} + \frac{q}{\ell(n)}\right) \\
    & \leq q\cdot(\rep+1)^2 \cdot \left(\frac{\Delta}{2^{r(n)}} + \frac{q}{\ell(n)}\right) \\
    & \underbrace{\leq}_{\text{\cref{eq:hard9.333:n}}}
        \frac{\delta}{k} \cdot \left(1-\frac{\bal+\delta}{q}\right),
\end{aligned}
\end{align}
where the second inequality holds
because each proof location is queried by $\W$, $\A_q$, and $\A_{k-q\rep}$ with probability
at most $\frac{\Delta}{2^{r(n)}} + \frac{q}{\ell(n)}$
owing to \cref{clm:hard9.333:W}.

Consequently, we evaluate $\Vhorn$'s rejection probability on $\Pf^\circ$ as follows:
\begin{align}
\label{eq:hard9.333:Horn:reject}
\begin{aligned}
    & \Pr\bigl[ \Vhorn \text{ rejects } \Pf^\circ \bigr] \\
    & \geq 1 - \Pr_{R_1,\ldots,R_{\rep+1}}\bigl[
        I_1, \ldots, I_{\rep+1} \text{ are not pairwise disjoint}
    \bigr] \\
    & \quad\quad - \Pr_{R_1,\ldots,R_{\rep+1}}\Bigl[
        \bigl( D_1(\Pf^\circ|_{I_1}) = 1 \bigr) \vee
        \bigl( D_2(\Pf^\circ|_{I_2}) = 0 \bigr)
        \vee \cdots \vee
        \bigl( D_{\rep+1}(\Pf^\circ|_{I_{\rep+1}}) = 0 \bigr)
    \Bigr] \\
    & \geq 1
    - \frac{\delta}{k} \cdot \left(1-\frac{\bal+\delta}{q}\right)
    - \left(1 - \frac{\bal-\delta}{k} \cdot \left(1-\frac{\bal+\delta}{q}\right)\right) \\
    & = \frac{\bal-2\delta}{k} \cdot \left(1-\frac{\bal+\delta}{q}\right) \\
    & > \frac{1}{k} \cdot \left(\frac{4}{q}-\epsilon\right),
\end{aligned}
\end{align}
where the last inequality can be shown as follows:
\begin{align}
    (\bal-2\delta) \cdot \left(1-\frac{\bal+\delta}{q}\right)
    \underbrace{=}_{\bal = \frac{q}{2} \text{ and } \delta = \frac{\epsilon}{4}}
        \left(\frac{q}{2} - \frac{\epsilon}{2}\right) \cdot \left(\frac{1}{2} - \frac{\epsilon}{4q}\right)
    = \frac{q}{4} - \frac{\epsilon}{4} - \frac{\epsilon}{8} + \frac{\epsilon^2}{8 \cdot q}
    > \frac{q}{4}-\epsilon,
\end{align}
which completes the proof.
\end{proof}

\begin{remark}
The choice of $\bal$ comes from the fact that
assuming that $\delta = 0$,
the second-to-last line of \cref{eq:hard9.333:Horn:reject}
is maximized when $\bal = \frac{q}{2}$\textup{;} namely,
\begin{align}
\begin{aligned}
    \frac{\partial}{\partial \bal} \left( \frac{\bal}{k} \cdot \left(1-\frac{\bal}{q}\right) \right) = 0
    \implies \bal = \frac{q}{2}.
\end{aligned}
\end{align}
\end{remark}

\subsubsection{Emulating the Horn Verifier}

Here, we emulate the Horn verifier $\Vhorn$ by an E$k$-CNF formula.
Recall that $\Vhorn$'s acceptance condition is the following:
\begin{align}
\label{eq:hard9.333:SAT:pos}
    \bigl( D_1(\Pf|_{I_1}) = 1 \bigr) \vee
    \bigl( D_2(\Pf|_{I_2}) = 0 \bigr)
    \vee \cdots \vee
    \bigl( D_{\rep+1}(\Pf|_{I_{\rep+1}}) = 0 \bigr).
\end{align}
Consider first the following $k$-query \emph{\scOR-predicate} verifier $\X$ obtained by modifying $\Vhorn$.
\begin{itembox}[l]{\textbf{$k$-query \scOR-predicate verifier $\X$ emulating $\Vhorn$}}
\begin{algorithmic}[1]
    \item[\textbf{Input:}]
        the $k$-query Horn verifier $\Vhorn$ and
        an input $x \in \zo^n$.
    \item[\textbf{Oracle access:}]
        a proof $\Pf \in \zo^{2\ell(n)}$.
    \State run $\Vhorn$ to generate
        $\rep+1$ query sequences $I_1, \ldots, I_{\rep+1}$ and
        $\rep+1$ circuits $D_1, \ldots, D_{\rep+1}$.
    \If{$I_1, \ldots, I_{\rep+1}$ are not pairwise disjoint}
        \State \Return $1$.
    \EndIf
    \State let $I \defeq \bigcup_{1 \leq i \leq {\rep+1}} I_i$.  \Comment{$|I| = k$.}
    \State sample a partial proof
    $\tilde{\Pf} \in \zo^I$
    that violates \cref{eq:hard9.333:SAT:pos} uniformly at random; namely,
    \begin{align}
    \label{eq:hard9.333:SAT:neg}
        \bigl( D_1(\tilde{\Pf}|_{I_1}) = 0 \bigr) \wedge
        \bigl( D_2(\tilde{\Pf}|_{I_2}) = 1 \bigr) \wedge \cdots \wedge
        \bigl( D_{\rep+1}(\tilde{\Pf}|_{I_{\rep+1}}) = 1 \bigr).\vspace{1cm}
    \end{align}
    \If{$\Pf|_I \neq \tilde{\Pf}$}  
        \State \Return $1$.
    \Else
        \State \Return $0$.
    \EndIf
\end{algorithmic}
\end{itembox}

\noindent
Note that $\X$ queries \emph{exactly} $k$ locations of the proof $\Pf$.
Since $D_1$ rejects at most $2^q-1$ strings,
$D_2, \ldots, D_{\rep}$ accept only $1^q$, and
$D_{\rep+1}$ accepts only $1^{k-q\rep}$,
the number of partial proofs $\tilde{\Pf} \in \zo^I$
such that \cref{eq:hard9.333:SAT:neg} holds is 
\begin{align}
    \bigl| D_1^{-1}(0) \times D_2^{-1}(1) \times \cdots \times D_{\rep+1}^{-1}(1) \bigr|
    = \bigl|D_1^{-1}(0)\bigr| \cdot \underbrace{1 \cdots 1}_{\rep \text{ times}}
    \leq 2^q-1,
\end{align}
where $D^{-1}(b) \defeq \{f \in \zo^q \mid D(f) = b\}$ for a circuit $D \colon \zo^q \to \zo$.
Conditioned on the event that \cref{eq:hard9.333:SAT:pos} does not hold,
$\X$ rejects $\Pf$ with probability at least $\frac{1}{2^q-1}$.
One can emulate $\X$ by an E$k$-CNF formula $\phi$ generated by the following procedure.

\begin{itembox}[l]{\textbf{Construction of an E$k$-CNF formula $\phi$ emulating $\X$}}
\begin{algorithmic}[1]
    \item[\textbf{Input:}]
        the $k$-query Horn verifier $\Vhorn$ and
        an input $x \in \zo^n$.
    \State let $\phi$ be an empty formula over $2\ell(n)$ variables, denoted by $x_1, \ldots, x_{2\ell(n)}$.
    \For{\textbf{each} random bit string $R \in \zo^{\Theta(\log n)}$ used by $\Vhorn$}
        \State run $\Vhorn$ on $R$ to generate
        $\rep+1$ query sequences $I_1, \ldots, I_{\rep+1}$ and
        $\rep+1$ circuits $D_1, \ldots, D_{\rep+1}$.
        \If{$I_1, \ldots, I_{\rep+1}$ are pairwise disjoint}
            \State let $I \defeq \bigcup_{1 \leq i \leq {\rep+1}} I_i$.  \Comment{$|I| = k$.}
            \For{\textbf{each} partial proof $\tilde{\Pf} \in \zo^I$
                such that \cref{eq:hard9.333:SAT:neg} holds}
                \label{linum:hard9.333:SAT:start}
                \LComment{there are at most $2^q-1$ partial proofs $\tilde{\Pf}$ in total.}
                \State generate a clause $C_{\tilde{\Pf}}$ that enforces $(x_i)_{i \in I} \neq \tilde{\Pf}$; namely,
                \begin{align}
                    C_{\tilde{\Pf}} \defeq \bigvee_{i \in I} \bigl\llbracket x_i \neq \tilde{\Pf}(i) \bigr\rrbracket,
                    \text{ where }
                    \bigl\llbracket x_i \neq \tilde{\Pf}(i) \bigr\rrbracket \defeq
                    \begin{cases}
                        x_i & \text{if } \tilde{\Pf}(i) = 0, \\
                        \bar{x_i} & \text{if } \tilde{\Pf}(i) = 1.
                    \end{cases}
                \end{align}
                \State add $C_{\tilde{\Pf}}$ into $\phi$.
                \label{linum:hard9.333:SAT:end}
            \EndFor
        \EndIf
    \EndFor
    \State \textbf{return} $\phi$.
\end{algorithmic}
\end{itembox}

\noindent
The above construction of $\phi$ runs in polynomial time in $n$.

We are now ready to complete the proof of \cref{lem:hard9.333:main}.
\begin{proof}[Proof of \cref{lem:hard9.333:main}]
We first show the completeness.
Suppose that $x \in L$; i.e., $\opt_{\V_L(x)}(\pf_\sss \reco \pf_\ttt) = 1$.
By \cref{lem:hard9.333:Horn},
we have $\opt_{\Vhorn}(\Pf_\sss \reco \Pf_\ttt) = 1$,
which implies that
$\opt_\phi(\Pf_\sss \reco \Pf_\ttt) = 1$
due to the construction of $\phi$.

We next show the soundness.
Suppose that $x \notin L$; i.e., $\opt_{\V_L(x)}(\pf_\sss \reco \pf_\ttt) < 1-\gap$.
Let $\sq{\Pf}$ be any reconfiguration sequence from $\Pf_\sss$ to $\Pf_\ttt$
such that $\val_\phi(\sq{\Pf}) = \opt_\phi(\Pf_\sss \reco \Pf_\ttt)$.
By \cref{lem:hard9.333:Horn},
$\sq{\Pf}$ contains a proof $\Pf^\circ$ such that
\begin{align}
    \Pr\bigl[\Vhorn \text{ rejects } \Pf^\circ\bigr]
    > \frac{1}{k} \cdot \left(\frac{q}{4}-\epsilon\right).
\end{align}
Conditioned on $I_1, \ldots, I_{\rep+1}$ and $D_1, \ldots, D_{\rep+1}$
such that \cref{eq:hard9.333:SAT:pos} does not hold on $\Pf^\circ$,
we have
\begin{align}
    \Pr_{\tilde{\Pf}}\bigl[
        \Pf^\circ|_I = \tilde{\Pf} \bigm|
        I_1, \ldots, I_{\rep+1}, D_1, \ldots, D_{\rep+1},\;
        \text{and \cref{eq:hard9.333:SAT:pos} does not hold}
    \bigr]
    \geq \frac{1}{2^q-1}.
\end{align}
Therefore,
exactly one of the (at most) $2^q-1$ clauses
generated in lines~\ref{linum:hard9.333:SAT:start}--\ref{linum:hard9.333:SAT:end}
of the construction of $\phi$
must be violated by $\Pf^\circ$.
Consequently, we derive
\begin{align}
    & 1 - \val_\phi(\Pf^\circ)
    \geq \Pr\bigl[\Vhorn \text{ rejects } \Pf^\circ\bigr] \cdot \frac{1}{2^q-1}
    > \frac{1}{(2^q-1) \cdot k} \cdot \left(\frac{q}{4}-\epsilon\right) \\
    & \implies \opt_\phi\bigl(\Pf_\sss \reco \Pf_\ttt\bigr)
    < 1 - \frac{1}{(2^q-1) \cdot k} \cdot \left(\frac{q}{4}-\epsilon\right),
\end{align}
which accomplishes the proof.
\end{proof}

\appendix

\section{$\NP$-hardness of \texorpdfstring{$\left(1-\frac{1}{8k}\right)$-factor}{(1-1/8k)-factor} Approximation}
\label{sec:NP}

Here, we give a simple proof that \MMkSATReconf is $\NP$-hard to approximate within a factor of $1-\frac{1}{8k}$ for every $k \geq 3$.

\begin{theorem}
\label{thm:NP}
For any integer $k \geq 3$,
\prb{Gap$_{1,1-\frac{1}{8k}}$ \kSATReconf} is $\NP$-hard.
In particular,
\MMkSATReconf is $\NP$-hard to approximate within a factor of $1-\frac{1}{8k}$ for every integer $k \geq 3$.
\end{theorem}

To prove \cref{thm:NP},
we first present a gap-preserving reduction from
\prb{Max E3-SAT} to \MMkSATReconf for every $k \geq 5$,
which is based on \cite[Theorem~5]{ito2011complexity}.

\begin{lemma}
\label{lem:NP:kSATReconf}
For any integer $k \geq 5$ and
any real $\delta > 0$,
there exists a polynomial-time reduction from
\prb{Gap$_{1,1-\delta}$ E3-SAT} to
\prb{Gap$_{1,1-\frac{\delta}{k-3}}$ \kSATReconf}.
Therefore,
\prb{Gap$_{1,1-\frac{1-\epsilon}{8(k-3)}}$ \kSATReconf} is $\NP$-hard
for any real $\epsilon > 0$.
\end{lemma}
\begin{proof}
Let $k \geq 5$ be an integer and
$\phi$ be an E$3$-CNF formula 
consisting of $m$ clauses $C_1, \ldots, C_m$ over $n$ variables $x_1, \ldots, x_n$.
We construct an instance $(\psi, \asg_\sss, \asg_\ttt)$ of \MMkSATReconf
as follows.
Define $K \defeq k-3 \geq 2$.
Create $K$ fresh variables $y_1, \ldots, y_K$.
Let $H_1, \ldots, H_K$ denote the $K$ possible Horn clauses (with a single positive literal)
over $y_1, \ldots, y_K$; namely,
\begin{align}
    H_i \defeq \bar{y_1} \vee \cdots \vee \bar{y_{i-1}} \vee y_i \vee \bar{y_{i+1}} \vee \cdots \vee \bar{y_K}.
\end{align}
Starting from an empty formula $\psi$ over $x_1, \ldots, x_n, y_1, \ldots, y_K$,
for each clause $C_j$ of $\phi$ and
each Horn clause $H_i$,
we add $C_j \vee H_i$ to $\psi$.
Note that $\psi$ contains $Km$ clauses.
The starting and ending assignments are defined as
$\asg_\sss \defeq 1^{n+K}$ and
$\asg_\ttt \defeq 0^{n+K}$, respectively.
Since every clause of $\psi$ contains both positive and negative literals,
both $\asg_\sss$ and $\asg_\ttt$ satisfy $\psi$,
completing the description of the reduction.

We first show the completeness;
i.e., $\exists \asg, \val_\phi(\asg) = 1$ implies
$\opt_{\psi}(\asg_\sss \reco \asg_\ttt) = 1$.
Consider a reconfiguration sequence $\sq{\asg}$ 
from $\asg_\sss$ to $\asg_\ttt$ obtained by the following procedure.
\begin{itembox}[l]{\textbf{Reconfiguration sequence $\sq{\asg}$ from $\asg_\sss$ to $\asg_\ttt$}}
\begin{algorithmic}[1]
    \State let $\asg^* \colon \{x_1, \ldots, x_n\} \to \zo$ be a satisfying assignment of $\phi$.
    \LComment{start with $\asg_\sss$.}
    \For{\textbf{each} variable $x_i$}
        \If{$\asg_\sss(x_i) \neq \asg^*(x_i)$}
            \State flip $x_i$'s current assignment from $\asg_\sss(x_i)$ to $\asg^*(x_i)$.
        \EndIf
    \EndFor
    \LComment{the current assignment to $\{x_1, \ldots, x_n\}$ is equal to $\asg^*$.}
    \For{\textbf{each} variable $y_i$}
        \State flip $y_i$'s current assignment from $1$ to $0$.
    \EndFor
    \For{\textbf{each} variable $x_i$}
        \If{$\asg^*(x_i) \neq \asg_\ttt(x_i)$}
            \State flip $x_i$'s current assignment from $\asg^*(x_i)$ to $\asg_\ttt(x_i)$.
        \EndIf
    \EndFor
    \LComment{end with $\asg_\ttt$.}
\end{algorithmic}
\end{itembox}

\noindent
For any intermediate assignment $\asg^\circ$ of $\sq{\asg}$,
it holds that either
$\asg^\circ|_{\{x_1, \ldots, x_n\}} = \asg^*$,
$\asg^\circ|_{\{y_1, \ldots, y_K\}} = 1^K$, or
$\asg^\circ|_{\{y_1, \ldots, y_K\}} = 0^K$;
thus, $\asg^\circ$ satisfies $\psi$, implying that
$\opt_{\psi}(\asg_\sss \reco \asg_\ttt) \geq \val_{\psi}(\sq{\asg}) = 1$, as desired.

We then show the soundness; i.e.,
$\forall \asg, \val_\phi(\asg) < 1-\delta$ implies
$\opt_{\psi}(\asg_\sss \reco \asg_\ttt) < 1 - \frac{\delta}{K}$.
Let $\sq{\asg} = (\asg^{(1)}, \ldots, \asg^{(T)})$ be
any reconfiguration sequence from $\asg_\sss$ to $\asg_\ttt$.
There must exist an assignment $\asg^\circ$ in $\sq{\asg}$ such that
$\asg^\circ|_{\{y_1, \ldots, y_K\}}$ contains a single $0$.
Let $i^\star \in [K]$ be a unique index such that
$\asg^\circ(y_{i^\star}) = 0$ and
$\asg^\circ(y_i) = 1$ for every $i \neq i^\star$.
By construction,
$\asg^\circ$ may not satisfy a clause $C_j \vee H_{i^\star}$
whenever $\asg^\circ|_{\{x_1, \ldots, x_n\}}$ does not satisfy a clause $C_j$.
Consequently, 
$\asg^\circ$ violates more than $\delta m$ clauses of $\psi$, implying that
\begin{align}
\val_{\psi}(\sq{\asg})
\leq \val_\psi(\asg^\circ)
< \frac{K m - \delta m}{K m}
= 1 - \frac{\delta}{k-3},
\end{align}
as desired.
The $\NP$-hardness of 
\prb{Gap$_{1,1-\frac{1-\epsilon}{8(k-3)}}$ E$k$-SAT Reconfiguration}
follows from
that of 
\prb{Gap$_{1,1-\frac{1-\epsilon}{8}}$ E3-SAT}
for any real $\epsilon > 0$ due to \cite[Theorem~6.5]{hastad2001some}.
\end{proof}

Since the above reduction does not work when $k \leq 4$,
the subsequent lemmas separately give a gap-preserving reduction from
\prb{Max E3-SAT} to \MMtreSATReconf and \prb{Maxmin E4-SAT Reconfiguration}.

\begin{lemma}
\label{lem:NP:3SATReconf}
    \prb{Gap$_{1,\frac{19}{20} + \epsilon}$ \treSATReconf}
    is $\NP$-hard
    for any real $\epsilon > 0$.
\end{lemma}

\begin{lemma}
\label{lem:NP:4SATReconf}
    \prb{Gap$_{1,\frac{10}{11} + \epsilon}$ E4-SAT Reconfiguration}
    is $\NP$-hard
    for any real $\epsilon > 0$.
\end{lemma}

The proof of \cref{thm:NP} is now immediate from
\cref{lem:NP:kSATReconf,lem:NP:3SATReconf,lem:NP:4SATReconf}.
\begin{proof}[Proof of \cref{thm:NP}]
The following hold, as desired:
\begin{itemize}
\item 
    By \cref{lem:NP:3SATReconf},
    \prb{Gap$_{1,1-\frac{1}{8\cdot 3}}$ \treSATReconf} is $\NP$-hard.
\item
    By \cref{lem:NP:4SATReconf},
    \prb{Gap$_{1,1-\frac{1}{8\cdot 4}}$ E4-SAT Reconfiguration} is $\NP$-hard.
\item
    Substituting $\epsilon$ of \cref{lem:NP:kSATReconf} by $\frac{3}{k}$ derives that
    \prb{Gap$_{1,1-\frac{1}{8k}}$ \kSATReconf} is $\NP$-hard for each integer $k \geq 5$.
    \qedhere
\end{itemize}
\end{proof}

\begin{proof}[Proof of \cref{lem:NP:3SATReconf}]
We first demonstrate a gap-preserving reduction from
\prb{Gap$_{1,1-\delta}$ E3-SAT}
to
\prb{Gap$_{1,\frac{\delta}{1+2\delta}}$ $\{3,4\}$-SAT Reconfiguration},
where ``$\{3,4\}$-SAT'' means that each clause has width $3$ or $4$.
Let $\phi$ be an E$3$-CNF formula consisting of
$m$ clauses $C_1, \ldots, C_m$ over $n$ variables $x_1, \ldots, x_n$.
We construct an instance $(\psi, \asg_\sss, \asg_\ttt)$
of \prb{Maxmin $\{3,4\}$-SAT Reconfiguration} as follows.
Create a CNF formula $\psi$ by the following procedure,
which is parameterized by $m_1$ and $m_2$.
\begin{itembox}[l]{\textbf{Construction of $\psi$}}
\begin{algorithmic}[1]
    \State introduce three fresh variables, denoted by $y$, $z_1$, and $z_2$.
    \State let $\psi$ be an empty formula over $n+3$ variables $x_1, \ldots, x_n, y, z_1, z_2$.
    \For{\textbf{each} $1 \leq j \leq m$}
        \State add a new clause $C_j \vee y$ to $\psi$.
    \EndFor
    \State add $m_1$ copies of a clause $\bar{y} \vee z_1 \vee \bar{z_2}$ to $\psi$.
    \State add $m_2$ copies of a clause $\bar{y} \vee \bar{z_1} \vee z_2$ to $\psi$.
\end{algorithmic}
\end{itembox}

\noindent
Note that $\psi$ consists of $m+m_1+m_2$ clauses, each of which has width $3$ or $4$.
The starting and ending assignments,
denoted by $\asg_\sss, \asg_\ttt \colon \{x_1, \ldots, x_n, y, z_1, z_2\} \to \zo$,
are defined as follows:
\begin{itemize}
    \item $\asg_\sss(x_i) \defeq 1$ for every $i \in [n]$ and
    $\asg_\sss(y, z_1, z_2) \defeq (1,1,1)$;
    \item $\asg_\ttt(x_i) \defeq 0$ for every $i \in [n]$ and
    $\asg_\ttt(y, z_1, z_2) \defeq (1,0,0)$.
\end{itemize}
Since $\asg_\sss$ and $\asg_\ttt$ satisfy $\psi$,
this completes the description of the reduction.

We first show the completeness; i.e.,
$\exists \asg, \val_\phi(\asg) = 1$ implies $\opt_{\psi}(\asg_\sss \reco \asg_\ttt) = 1$.
Consider a reconfiguration sequence $\sq{\asg}$ from $\asg_\sss$ to $\asg_\ttt$
obtained by the following procedure.
\begin{itembox}[l]{\textbf{Reconfiguration sequence $\sq{\asg}$ from $\asg_\sss$ to $\asg_\ttt$}}
\begin{algorithmic}[1]
    \State let $\asg^* \colon \{x_1, \ldots, x_n\} \to \{0,1\}$
    be a satisfying assignment of $\phi$.
    \LComment{start with $\asg_\sss$.}
    \For{\textbf{each} variable $x_i$}
        \If{$\asg_\sss(x_i) \neq \asg^*(x_i)$}
            \State flip $x_i$'s current assignment from $\asg_\sss(x_i)$ to $\asg^*(x_i)$.
        \EndIf
    \EndFor
    \State
    flip the assignment to $y$, $z_1$, $z_2$, and $y$ in this order.
    \LComment{the above step gives rise to the following reconfiguration sequence of assignments to $(y,z_1,z_2)$\textup{:}
    $((1,1,1), (0,1,1), (0,0,1), (0,0,0), (1,0,0))$.}
    \For{\textbf{each} variable $x_i$}
        \If{$\asg^*(x_i) \neq \asg_\ttt(x_i)$}
            \State flip $x_i$'s current assignment from $\asg^*(x_i)$ to $\asg_\ttt(x_i)$.
        \EndIf
    \EndFor
    \LComment{end with $\asg_\ttt$.}
\end{algorithmic}
\end{itembox}

\noindent
For any intermediate assignment $\asg^\circ$ of $\sq{\asg}$, the following hold:
\begin{itemize}
    \item Since $\asg^\circ(y) = 1$ or $\asg^\circ|_{\{x_1, \ldots, x_n\}} = \asg^*$,
        each clause $C_j \vee y$ is satisfied.
    \item Since $\asg^\circ(y,z_1,z_2) \neq (1,0,1)$,
         a clause $\bar{y} \vee z_1 \vee \bar{z_2}$ is satisfied.
    \item Since $\asg^\circ(y,z_1,z_2) \neq (1,1,0)$,
         a clause $\bar{y} \vee \bar{z_1} \vee z_2$ is satisfied.
\end{itemize}
Therefore, $\sq{\asg}$ satisfies $\psi$; i.e., $\opt_\psi(\asg_\sss \reco \asg_\ttt) = 1$.

We then show the soundness; i.e.,
$\forall \asg, \val_\phi(\asg) \leq 1-\delta$ implies
$\opt_{\psi}(\asg_\sss \reco \asg_\ttt) \leq 1-\frac{\delta}{1+2\delta}$.
Let $\sq{\asg} = (\asg^{(1)}, \ldots, \asg^{(T)})$
be any reconfiguration sequence from $\asg_\sss$ to $\asg_\ttt$.
We bound its value by the following case analysis:
\begin{description}
    \item[(Case 1)] $\exists t, \asg^{(t)}(y) = 0$. \\
        Each clause $C_j \vee y$ of $\psi$ is satisfied by $\asg^{(t)}$
        if and only if $C_j$ is satisfied by $\asg^{(t)}|_{\{x_1, \ldots, x_n\}}$.
        By assumption, at least $\delta m$ clauses of $\psi$
        must be unsatisfied by such $\asg^{(t)}$.
    \item[(Case 2)] $\forall t, \asg^{(t)}(y) = 1$. \\
        Since
        $\asg^{(1)}(y,z_1,z_2) = (1,1,1)$ and
        $\asg^{(T)}(y,z_1,z_2) = (1,0,0)$,
        there is some assignment $\asg^\circ$ in $\sq{\asg}$ such that
        $\asg^\circ(y,z_1,z_2)$ is
        $(1,0,1)$ or $(1,1,0)$.
        In the former case,
        $\asg^\circ$ violates
        $m_1$ clauses in the form of $\bar{y} \vee z_1 \vee \bar{z_2}$;
        in the latter case,
        $\asg^\circ$ violates
        $m_2$ clauses in the form of $\bar{y} \vee \bar{z_1} \vee z_2$.
\end{description}
In either case,
the maximum number of clauses violated by $\sq{\asg}$ must be at least
\begin{align}
    \min\bigl\{ \delta m, m_1, m_2 \bigr\}.
\end{align}
Letting $m_1 \defeq \delta m$ and
$m_2 \defeq \delta m$, we have
\begin{align}
\begin{aligned}
    \val_{\psi}(\sq{\asg})
    \leq 1 - \frac{\min\bigl\{ \delta m, m_1, m_2 \bigr\}}{m+m_1+m_2}
    = 1 - \frac{\delta}{1+2 \delta},
\end{aligned}
\end{align}
as desired.

Since \prb{Gap$_{1,1-\frac{1-\epsilon}{8}}$ E3-SAT} is $\NP$-hard
for any real $\epsilon > 0$ \cite[Theorem~6.5]{hastad2001some},
we let $\delta \defeq \frac{1-\epsilon}{8}$ to have that
\prb{Gap$_{1,1-\frac{\delta}{1+2\delta}}$ $\{3,4\}$-SAT Reconfiguration} is $\NP$-hard,
where $1-\frac{\delta}{1+2\delta}$ is bounded as follows:
\begin{align}
    1 - \frac{\delta}{1+2\delta}
    = 1 - \frac{1-\epsilon}{10-2\epsilon}
    \leq 1 - \frac{1-\epsilon}{10}.
\end{align}
By \cite{ohsaka2023gap},
\prb{Gap$_{1,1-\frac{1-\epsilon}{10}}$ $\{3,4\}$-SAT Reconfiguration}
is further reduced to
\prb{Gap$_{1,1-\frac{1-\epsilon}{20}}$ \treSATReconf}
in polynomial time,
which completes the proof.
\end{proof}

\begin{proof}[Proof of \cref{lem:NP:4SATReconf}]
We demonstrate a gap-preserving reduction from
\prb{Gap$_{1,1-\delta}$ E$3$-SAT}
to
\prb{Gap$_{1,\frac{\delta}{1+3\delta}}$ E4-SAT Reconfiguration}.
Let $\phi$ be an E$3$-CNF formula
consisting of $m$ clauses $C_1, \ldots, C_m$ over $n$ variables $x_1, \ldots, x_n$.
We construct an instance
$(\psi, \asg_\sss, \asg_\ttt)$
of \prb{Maxmin E4-SAT Reconfiguration} as follows.
Create a CNF formula $\psi$ by the following procedure,
which is parameterized by $m_1$, $m_2$, and $m_3$.
\begin{itembox}[l]{\textbf{Construction of $\psi$}}
\begin{algorithmic}[1]
    \State introduce four fresh variables, denoted by $y$, $z_1$, $z_2$, and $z_3$.
    \State let $\psi$ be an empty formula over $n+4$ variables $x_1, \ldots, x_n, y, z_1, z_2, z_3$.
    \For{\textbf{each} $1 \leq j \leq m$}
        \State add a new clause $C_j \vee y$ to $\psi$.
    \EndFor
    \State add $m_1$ copies of a clause $\bar{y} \vee z_1 \vee \bar{z_2} \vee \bar{z_3}$ to $\psi$.
    \State add $m_2$ copies of a clause $\bar{y} \vee \bar{z_1} \vee z_2 \vee \bar{z_2}$ to $\psi$.
    \State add $m_3$ copies of a clause $\bar{y} \vee \bar{z_1} \vee \bar{z_2} \vee z_3$ to $\psi$.
\end{algorithmic}
\end{itembox}

\noindent
Note that $\psi$ consists of $m+m_1+m_2+m_3$ clauses of width $4$.
The starting and ending assignments, 
denoted by $\asg_\sss, \asg_\ttt \colon \{x_1, \ldots, x_n, y, z_1, z_2, z_3\} \to \zo$,
are defined as follows:
\begin{itemize}
    \item $\asg_\sss(x_i) \defeq 1$ for every $i \in [n]$ and
    $\asg_\sss(y, z_1, z_2, z_3) \defeq (1,1,1,1)$;
    \item $\asg_\ttt(x_i) \defeq 0$ for every $i \in [n]$ and
    $\asg_\ttt(y, z_1, z_2, z_3) \defeq (1,0,0,0)$.
\end{itemize}
Since $\asg_\sss$ and $\asg_\ttt$ satisfy $\psi$,
this completes the description of the reduction.

We first show the completeness; i.e.,
$\exists \asg, \val_\phi(\asg) = 1$ implies $\opt_{\psi}(\asg_\sss \reco \asg_\ttt) = 1$.
Consider a reconfiguration sequence $\sq{\asg}$ from $\asg_\sss$ to $\asg_\ttt$ obtained by the following procedure.
\begin{itembox}[l]{\textbf{Reconfiguration sequence $\sq{\asg}$ from $\asg_\sss$ to $\asg_\ttt$}}
\begin{algorithmic}[1]
    \State let $\asg^* \colon \{x_1, \ldots, x_n\} \to \{0,1\}$
    be a satisfying assignment of $\phi$.
    \LComment{start with $\asg_\sss$.}
    \For{\textbf{each} variable $x_i$}
        \If{$\asg_\sss(x_i) \neq \asg^*(x_i)$}
            \State flip $x_i$'s current assignment from $\asg_\sss(x_i)$ to $\asg^*(x_i)$.
        \EndIf
    \EndFor
    \State
    flip the assignment to $y$, $z_1$, $z_2$, $z_3$, and $y$ in this order.
    \LComment{the above step gives rise to the following reconfiguration sequence of assignments to $(y,z_1,z_2,z_3)$\textup{:}
    $((1,1,1,1), (0,1,1,1), (0,0,1,1), (0,0,0,1), (0,0,0,0), (1,0,0,0))$.}
    \For{\textbf{each} variable $x_i$}
        \If{$\asg^*(x_i) \neq \asg_\ttt(x_i)$}
            \State flip $x_i$'s current assignment from $\asg^*(x_i)$ to $\asg_\ttt(x_i)$.
        \EndIf
    \EndFor
    \LComment{end with $\asg_\ttt$.}
\end{algorithmic}
\end{itembox}

\noindent
For any intermediate assignment $\asg^\circ$ of $\sq{\asg}$, the following hold:
\begin{itemize}
    \item Since $\asg^\circ(y) = 1$ or $\asg^\circ|_{\{x_1, \ldots, x_n\}} = \asg^*$,
        each clause $C_j \vee y$ is satisfied.
    \item Since $\asg^\circ(y,z_1,z_2,z_3) \neq (1,0,1,1)$,
         a clause $\bar{y} \vee z_1 \vee \bar{z_2} \vee \bar{z_3}$ is satisfied.
    \item Since $\asg^\circ(y,z_1,z_2,z_3) \neq (1,1,0,1)$,
         a clause $\bar{y} \vee \bar{z_1} \vee z_2 \vee \bar{z_3}$ is satisfied.
    \item Since $\asg^\circ(y,z_1,z_2,z_3) \neq (1,1,1,0)$,
         a clause $\bar{y} \vee \bar{z_1} \vee \bar{z_2} \vee z_3$ is satisfied.
\end{itemize}
Therefore, $\sq{\asg}$ satisfies $\psi$; i.e., $\opt_\psi(\asg_\sss \reco \asg_\ttt) = 1$.

We then show the soundness; i.e.,
$\forall \asg, \val_\phi(\asg) \leq 1-\delta$ implies
$\opt_{\psi}(\asg_\sss \reco \asg_\ttt) \leq 1-\frac{\delta}{1+3\delta}$.
Let $\sq{\asg} = (\asg^{(1)}, \ldots, \asg^{(T)})$
be any reconfiguration sequence from $\asg_\sss$ to $\asg_\ttt$.
We bound its value by the following case analysis:
\begin{description}
    \item[(Case 1)] $\exists t, \asg^{(t)}(y) = 0$. \\
        Each clause $C_j \vee y$ of $\psi$ is satisfied by $\asg^{(t)}$
        if and only if $C_j$ is satisfied by $\asg^{(t)}|_{\{x_1, \ldots, x_n\}}$.
        By assumption, at least $\delta m$ clauses of $\psi$
        must be unsatisfied by such $\asg^{(t)}$.
    \item[(Case 2)] $\forall t, \asg^{(t)}(y) = 1$. \\
        Since
        $\asg^{(1)}(y,z_1,z_2,z_3) = (1,1,1,1)$ and
        $\asg^{(T)}(y,z_1,z_2,z_3) = (1,0,0,0)$,
        there is some assignment $\asg^\circ$ in $\sq{\asg}$ such that
        $\asg^\circ(y,z_1,z_2,z_3)$ is 
        $(1,0,1,1)$, $(1,1,0,1)$, or $(1,1,1,0)$.
        In the first case,
        $\asg^\circ$ violates
        $m_1$ clauses in the form of $\bar{y} \vee z_1 \vee \bar{z_2} \vee \bar{z_3}$;
        in the second case,
        $\asg^\circ$ violates
        $m_2$ clauses in the form of $\bar{y} \vee \bar{z_1} \vee z_2 \vee \bar{z_3}$;
        in the third case,
        $\asg^\circ$ violates
        $m_3$ clauses in the form of $\bar{y} \vee \bar{z_1} \vee \bar{z_2} \vee z_3$.
\end{description}
In either case,
the maximum number of clauses violated by $\sq{\asg}$ must be at least
\begin{align}
    \min\bigl\{ \delta m, m_1, m_2, m_3 \bigr\}.
\end{align}
Letting
$m_1 \defeq \delta m$,
$m_2 \defeq \delta m$, and
$m_3 \defeq \delta m$, we have
\begin{align}
\begin{aligned}
    \val_{\psi}(\sq{\asg})
    \leq 1 - \frac{\min\bigl\{ \delta m, m_1, m_2, m_3 \bigr\}}{m+m_1+m_2+m_3}
    = 1 - \frac{\delta}{1+3 \delta},
\end{aligned}
\end{align}
as desired.

Since \prb{Gap$_{1,1-\frac{1-\epsilon}{8}}$ E3-SAT} is $\NP$-hard
for any real $\epsilon > 0$ \cite[Theorem~6.5]{hastad2001some},
we let $\delta \defeq \frac{1-\epsilon}{8}$ to have that
\prb{Gap$_{1,1-\frac{\delta}{1+3\delta}}$ E4-SAT Reconfiguration} is $\NP$-hard,
where $1-\frac{\delta}{1+3\delta}$ is bounded as follows:
\begin{align}
    1 - \frac{\delta}{1+3\delta}
    = 1 - \frac{1-\epsilon}{11-3\epsilon}
    \leq 1 - \frac{1-\epsilon}{11},
\end{align}
which completes the proof.
\end{proof}

\section{Omitted Proofs}
\label{app}

\begin{proof}[Proof of \cref{fct:alg:binom}]
Using the fact that
\begin{align}
    \binom{n+1}{k+1} = \frac{n+1}{k+1} \binom{n}{k},
\end{align}
we have
\begin{align}
    \sum_{0 \leq k \leq n} \binom{n}{k} \frac{1}{k+1}
    = \sum_{0 \leq k \leq n} \binom{n+1}{k+1} \frac{1}{n+1}
    = \frac{1}{n+1} \sum_{0 \leq k \leq n} \binom{n+1}{k}
    = \frac{2^{n+1} - 1}{n+1}.
\end{align}
Similarly, we have
\begin{align}
\begin{aligned}
    \sum_{0 \leq k \leq n} \binom{n}{k} \frac{1}{k+2}
    & = \sum_{0 \leq k \leq n} \binom{n+1}{k+1} \frac{k+1}{n+1} \frac{1}{k+2} \\
    & \underbrace{=}_{\text{replace } k \text{ by } k-1} \frac{1}{n+1} \sum_{1 \leq k \leq n+1} \binom{n+1}{k} \left(1-\frac{1}{k+1}\right) \\
    & = \frac{1}{n+1} \Biggl[
        \underbrace{\sum_{0 \leq k \leq n+1} \binom{n+1}{k}}_{= 2^{n+1}}
        - \underbrace{\sum_{0 \leq k \leq n+1} \binom{n+1}{k} \frac{1}{k+1}}_{= \frac{2^{n+2}-1}{n+2}}
    \Biggr] \\
    & = \frac{2^{n+1}\cdot n + 1}{(n+1)(n+2)},
\end{aligned}
\end{align}
as desired.
\end{proof}

\begin{proof}[Proof of \cref{cor:hard0}]
To prove \cref{cor:hard0}, we use the following claim, which will be proven later.

\begin{claim}
\label{clm:hard:gamma}
For any integer $k \geq 3$,
any real $\gamma > 1$ with $\gamma k \in \bbN$, and
any real $\epsilon > 0$,
there exists a gap-preserving reduction from
\prb{Gap$_{1,1-\epsilon}$ E$(\gamma k)$-SAT Reconfiguration}
to
\prb{Gap$_{1,1-\frac{\epsilon}{\Gamma}}$ \kSATReconf},
where $\Gamma \defeq \left\lceil \frac{\gamma k}{k-2} \right\rceil$.
\end{claim}

Let $\epsilon \defeq 0.2$ and
$k_0(\epsilon) \in \bbN$ be an integer as defined in \cref{thm:hard9.333}.
For any integer $k \geq k_0(\epsilon)$,
\prb{Gap$_{1,1-\frac{1}{10k}}$ \kSATReconf} is $\PSPACE$-hard
by \cref{thm:hard9.333}.
For any integer $k$ with $3 \leq k < k_0(\epsilon)$,
we apply \cref{clm:hard:gamma} to reduce
\prb{Gap$_{1,1-\frac{1}{10\cdot k_0(\epsilon)}}$ E$(k_0(\epsilon))$-SAT Reconfiguration} to
\prb{Gap$_{1,1-\frac{1}{10\cdot k_0(\epsilon) \cdot \Gamma}}$ \kSATReconf}, where
\begin{align}
    \Gamma
    = \left\lceil \frac{k_0(\epsilon)}{k-2} \right\rceil
    \underbrace{\leq}_{k \geq 3} k_0(\epsilon).
\end{align}
Letting
\begin{align}
    \delta_0 \defeq \frac{1}{10\cdot k_0(\epsilon)^2},
\end{align}
we derive that
\prb{Gap$_{1,1-\frac{\delta_0}{k}}$ \kSATReconf}
is $\PSPACE$-hard for every integer $k \geq 3$,
as desired.
\end{proof}

\begin{proof}[Proof of \cref{clm:hard:gamma}]
Let $(\phi,\asg_\sss,\asg_\ttt)$ be an instance of \prb{Maxmin E$(\gamma k)$-SAT Reconfiguration},
where $\phi$ is an E$(\gamma k)$-CNF formula
consisting of $m$ clauses $C_1, \ldots, C_m$ over $n$ variables $x_1, \ldots, x_n$, and
$\asg_\sss,\asg_\ttt$ are satisfying assignments for $\phi$.
We construct an instance $(\psi,\bsg_\sss,\bsg_\ttt)$ of \MMkSATReconf as follows.
Let $\Gamma \defeq \left\lceil \frac{\gamma k}{k-2} \right\rceil$.
Starting from an empty clause $\psi$,
for each clause $C_j = \ell_1 \vee \cdots \vee \ell_{\gamma k}$ of $\phi$,
we add to $\psi$ the $\Gamma$ clauses of width $k$ generated by the following procedure.
\begin{itembox}[l]{\textbf{Construction of $\Gamma$ clauses from a clause $C_j = \ell_1 \vee \cdots \vee \ell_{\gamma k}$ of $\phi$}}
\begin{algorithmic}[1]
    \State create $\Gamma$ sets of literals, denoted by $S_1, \ldots, S_\Gamma$,
    such that the following hold:
    \begin{itemize}
        \item each set $S_i$ contains exactly $k-2$ literals of $C_j$;
        \item $S_1, \ldots, S_\Gamma$ cover the $\gamma k$ literals of $C_j$; namely,
            $S_1 \cup \cdots \cup S_\Gamma = \{\ell_1, \ldots, \ell_{\gamma k}\}$.
    \end{itemize}
    \LComment{
        such a family of $\Gamma$ sets always exists because 
        $\displaystyle \frac{\gamma k}{k-2} \leq \Gamma$.
    }
    \State append a single literal of $C_j$ (say $\ell_1$) to each of $S_1$ and $S_\Gamma$,
    so that $|S_1| = |S_\Gamma| = k-1$ and $|S_2| = \cdots = |S_{\Gamma-1}| = k-2$.
    \State introduce $\Gamma-1$ fresh variables, denoted by $y_{j,1}, \ldots, y_{j,\Gamma-1}$.
    \State generate $\Gamma$ clauses of width $k$ representing the following formulas:
    \begin{align}
    \begin{aligned}
        y_{j,1} & \implies \biggl( \bigvee_{\ell_i \in S_1} \ell_i \biggr), \\
        y_{j,2} & \implies \biggl(y_{j,1} \vee \bigvee_{\ell_i \in S_2} \ell_i \biggr), \\
        & \vdots \\
        y_{j,\Gamma-1} & \implies \biggl(y_{j,\Gamma-2} \vee \bigvee_{\ell_i \in S_{\Gamma-1}} \ell_i \biggr), \\
        1 & \implies \biggl(y_{j,\Gamma-1} \vee \bigvee_{\ell_i \in S_\Gamma} \ell_i \biggr).
    \end{aligned}
    \end{align}
\end{algorithmic}    
\end{itembox}

\noindent
For a satisfying assignment $\asg \colon \{x_1, \ldots, x_n\} \to \zo$ for $\phi$,
we consider an assignment
$\bsg \colon \{x_1, \ldots, x_n, \allowbreak y_{1,1}, \ldots, y_{m,\Gamma-1}\} \to \zo$
for $\psi$ such that
$\bsg|_{\{x_1, \ldots, x_n\}} = \asg_{\{x_1, \ldots, x_n\}}$ and
$\bsg(y_{j,i})$ for each variable $y_{j,i}$ is defined as follows:
\begin{align}
    \bsg(y_{j,i}) \defeq 
    \begin{cases}
        0 & \text{ if } i \leq i_j - 1, \\
        1 & \text{ if } i > i_j - 1,
    \end{cases}
\end{align}
where $i_j \in [\Gamma]$ is an index such that
some literal of $S_{i_j}$ is satisfied by $\asg$.
Construct the starting assignment $\bsg_\sss$ from $\asg_\sss$ and
the ending assignment $\bsg_\ttt$ from $\asg_\ttt$ according to this procedure.
Observe that both $\bsg_\sss$ and $\bsg_\ttt$ satisfy $\psi$,
which completes the description of the reduction.
Similarly to \cite[Lemma~3.5]{gopalan2009connectivity} and \cite[Claim~3.4]{ohsaka2023gap},
we have the following completeness and soundness, as desired:
\begin{itemize}
    \item (Completeness)
        If $\opt_\phi\bigl(\asg_\sss \reco \asg_\ttt\bigr) = 1$,
        then $\opt_\psi\bigl(\bsg_\sss \reco \bsg_\ttt\bigr) = 1$.
    \item (Soundness)
        If $\opt_\phi\bigl(\asg_\sss \reco \asg_\ttt\bigr) < 1 - \epsilon$,
        then $\opt_\psi\bigl(\bsg_\sss \reco \bsg_\ttt\bigr) < 1 - \frac{\epsilon}{\Gamma}$.
        \qedhere
\end{itemize}
\end{proof}

\emergencystretch=1em  
\printbibliography

\end{document}